\documentclass[review]{elsarticle}
\usepackage[left=100pt, right=100pt]{geometry}
\usepackage{lineno,hyperref}
\modulolinenumbers[5]
\usepackage[utf8]{inputenc}
\usepackage{graphicx}
\usepackage{mwe}
\usepackage{subfig}
\usepackage {float}
\usepackage{amssymb}
\usepackage{amsmath}
\usepackage{dsfont}
\usepackage{mathtools}
\usepackage{amsthm}
\usepackage[utf8]{inputenc}
\usepackage[T1]{fontenc}
\usepackage{lmodern}
\usepackage{ae}
\usepackage{enumerate}
\usepackage{epsfig}
\newtheorem{theorem}{Theorem}[section]
\newtheorem{lemma}{Lemma}
\newtheorem{corollary}{Corollary}

\newtheorem{remark}{Remark}
\newtheorem{definition}{Definition}

\journal{Journal of \LaTeX\ Templates}

\usepackage{algorithm,algpseudocode}

\begin{document}

\begin{frontmatter}
	
	\title{\textbf{ Stochastic Dynamics and Probability Analysis for a Generalized 
    Epidemic Model with Environmental Noise }}
	\author[label1]{ Brahim Boukanjime\corref{cor1}}
\cortext[cor1]{Corresponding author}
\ead{brahim.boukanjime@gmail.com}
\affiliation[label1]{organization={Laboratory of Mathematical Modeling and Economic Calculations, Hassan First University of Settat},
  country={Morocco}}

\author[label2]{Mohamed Maama\corref{cor2}}
\cortext[cor2]{Corresponding author}
\affiliation[label2]{organization={Applied Mathematics and Computational Science Program,  Computer, Electrical and Mathematical Sciences and Engineering Division, King Abdullah University of Science and Technology},
	country={Kingdom of Saudi Arabia}}
 \ead{maama.mohamed@gmail.com}
 
	\begin{abstract}
In this paper we consider a stochastic SEIQR (susceptible-exposed-infected-quarantined-recovered) epidemic model with a generalized incidence function. Using the Lyapunov method, we establish the existence and uniqueness of a global positive solution to the model, ensuring that it remains well-defined over time. Through the application of Young’s inequality and Chebyshev’s inequality, we demonstrate the concepts of stochastic ultimate boundedness and stochastic permanence, providing insights into the long-term behavior of the epidemic dynamics under random perturbations. Furthermore, we derive conditions for stochastic extinction, which describes scenarios where the epidemic may eventually die out, and V-geometric ergodicity, which indicates the rate at which the system's state converges to its equilibrium. Finally, we perform numerical simulations to verify our theoretical results and assess the model’s behavior under different parameters.

	\end{abstract}
	
\end{frontmatter}

\section{Introduction}
Historically, the spread of infectious diseases has posed a significant threat to humanity. In the absence of immunity, entire civilizations have been reduced to ruins and monuments that stand as reminders of the impact of epidemics. In the early 1500s, smallpox devastated the Inca population, wiping out $93\%$ of it during the Spanish conquest. Similarly, the Bubonic plague resulted in the deaths of $10.000$ people daily in Europe between $1346$ and $1350$, in recent times, the ongoing Coronavirus (COVID-19) caused a deep international crisis with 2.3 million deaths resulting in global health problems and changes in lifestyle. Since then, advancements in technology and medicine have played a crucial role in mitigating such disasters. However, scientists have noted that these advancements might enable pathogens to evolve and find new ways to spread. This concern has driven mathematicians and computational scientists to develop concrete models to understand and predict such real-life phenomena. Unlike other sciences, mathematical models do not undergo experimental validation, and conducting experiments on disease spread would be both impractical and unethical. Thus, creating models to describe the dynamics of epidemics is crucial for devising real-time strategies to manage and control outbreaks by anticipating their behavior and taking proactive measures.\\
Mathematical modeling emerges as a crucial tool in this context. It is a vital component of understanding disease outbreaks and making epidemiological predictions \cite{01,02,03,04,05}. Mathematical models represent complex systems and scenarios, enabling the analysis of real-world problems and forecasting future events. These models translate biological and epidemiological processes into mathematical frameworks, providing valuable insights into disease dynamics and assisting in the formulation of effective response strategies.
\\
In this regard, our paper aims to advance the understanding of disease dynamics through mathematical modeling. The paper is organized as follows: Section 2 introduces the model structure, providing a detailed description of the underlying mathematical framework. Section 3 explores stochastically ultimate boundedness and permanence, addressing the stability and long-term behavior of the model. Section 4 investigates stochastic extinction, focusing on the conditions under which the disease may die out. Section 5 examines V-geometric ergodicity, offering insights into the model’s long-term behavior. Finally, Section 6 presents numerical simulations to validate the theoretical results and illustrate their practical implications.

\section{The Model Structure}
In this study, we examine the SEIQR model, which assumes that individuals who are infected and symptomatic are placed in quarantine for treatment, whether in a hospital or outpatient treatment (possibly in the home). After their quarantine period, these individuals transition into the recovered category. The population is stratified into five distinct compartments following the SEIQR model framework: Susceptible $(S)$, Exposed $(E)$, Infectious $(I)$, Quarantined $(Q)$, and Recovered $(R)$. The dynamics of these compartments are governed by the following deterministic system of differential equations:
\begin{equation}
\begin{aligned}
\frac{dS}{dt} &= b - \beta S h(I) - (\eta + \mu + d_3) S,\\
\frac{dE}{dt} &= \beta S h(I) - (\lambda + \mu + d_2) E,\\
\frac{dI}{dt} &= \lambda E - (\mu + \nu + \gamma + d_1) I,\\
\frac{dQ}{dt} &= d_3 S + d_2 E + d_1 I - (\mu + \tau) Q,\\
\frac{dR}{dt} &= \eta S + \tau Q + \gamma I - \mu R,
\end{aligned}
\end{equation}
where all parameters are positive constants. The model parameters are defined as follows: $b$ is the birth rate, $\beta$ is the transmission rate, $\mu$ is the natural death rate, $\eta$ is the vaccination rate, and $d_3$, $d_2$, $d_1$ are the transfer rates from the Susceptible, Exposed, and Infectious compartments to the Quarantined compartment, respectively. The parameter $\lambda$ represents the rate at which exposed individuals become infectious, $\nu$ is the disease-induced death rate, $\gamma$ is the recovery rate of infectious individuals, and $\tau$ is the recovery rate of quarantined individuals. The transmission of the infection is modeled using a general incidence function $h(I)$, which is a non-negative, twice continuously differentiable function. This function satisfies the conditions $h(0) = 0$ and $h'(0) > 0$, with the ratio $\frac{h(I)}{I}$ being monotonically decreasing for $I \geq 0$, implying that $\frac{h(I)}{I} < h'(0)$ for $I > 0$.
\\
In reality, epidemic models are not immune to environmental fluctuations, which induce random variations in system parameters around their average values. These fluctuations reflect the unpredictable nature of the environment, which impacts the spread of the disease. Such stochastic influences have been rigorously explored in previous studies, such as those in \cite{07,08,09,general,010,011,012}. Acknowledging the significance of these fluctuations, we extend the deterministic model to a stochastic framework by incorporating environmental noise. The resulting stochastic model employs a generalized incidence function as introduced in \cite{general}, and is described by the following system of stochastic differential equations (SDEs):

\begin{equation}\label{sys1}
\left\{ 
\begin{aligned}
dS(t) &= \left[b - \beta S h(I) - (\eta + \mu + d_3) S\right]dt + \sigma_1 S dB_1(t),\\
dE(t) &= \left[\beta S h(I) - (\lambda + \mu + d_2) E\right]dt + \sigma_2 E dB_2(t),\\
dI(t) &= \left[\lambda E - (\mu + \nu + \gamma + d_1) I\right]dt + \sigma_3 I dB_3(t),\\
dQ(t) &= \left[d_3 S + d_2 E + d_1 I - (\mu + \tau) Q\right]dt + \sigma_4 Q dB_4(t),\\
dR(t) &= \left[\eta S + \tau Q + \gamma I - \mu R\right]dt + \sigma_5 R dB_5(t),
\end{aligned}
\right.
\end{equation}
where $B_i(t)$, for $i=1,\ldots,5$, are independent standard Brownian motions defined on a complete probability space $\left(\Omega, \mathcal{F},\lbrace\mathcal{F}_t\rbrace_{t\geq 0}, \mathbb{P}\right)$, with $\sigma_i^2>0$ representing the perturbation intensities corresponding to each compartment. These Brownian motions introduce stochastic perturbations into the model, capturing the random fluctuations in the population dynamics due to environmental variability.

This stochastic framework allows us to more realistically capture the unpredictable nature of epidemic spread in a varying environment, providing a more robust basis for understanding the impact of noise on disease dynamics. The model presented here serves as a foundation for further analysis, including the investigation of long-term behavior, stability, and the effects of stochastic perturbations on disease control strategies.

\section{Existence and uniqueness of the global positive solution}
As the solution of stochastic differential equation \eqref{sys1} has biological significance, it should be non-negative. The solution of model \eqref{sys1} may explode in a finite time, since the coefficients don’t satisfy the linear growth condition, even though they satisfy the local Lipschitz condition. Khasminskii and Mao gave the Lyapunov function argument, which is a powerful test for non-explosion of solutions without the linear growth condition and referred to as the Khasminskii–Mao theorem \cite{1,2}.

\begin{theorem}\label{mth}
		For any given initial value $(S(0),E(0),I(0),Q(0),R(0)) \in\mathbb{R}^5_+$, there exists  a unique solution $(S(t),E(t),I(t),Q(t),R(t))$ of the model \eqref{sys1} on $t\geq 0$ and the solution will remain in $\mathbb{R}^5_+$ with probability one.
 	
\end{theorem}
\begin{proof}
Consider the $\mathcal{C}^2$-function \~{V}: $\mathbb{R}_+^5\longrightarrow \mathbb{R}_+$ by
\begin{eqnarray*}
\tilde{V}(S,E,I,Q,R)=S+E+I+Q+R-\ln S -\ln E-\ln I-\ln Q-\ln R-5.
\end{eqnarray*}
By using the Lyapunov analysis method \cite{2}, we can prove that the solution of system \eqref{sys1} is positive and global. The proof of Theorem \ref{mth} is standard (see for example, the proof of Lemma 2.1 in \cite{3}) and hence is omitted here.

\end{proof}
In the following, we are going to define the notions of stochastically
ultimate boundedness and permanence, and we will subsequently show that the solution of system \eqref{sys1} verifies these properties.

\section{Stochastically ultimate boundedness and permanence}
In this section, we aim at studying the stochastic permanence of model \eqref{sys1}. First of all, we present the following definitions.
\begin{definition}[Stochastically ultimately bounded \cite{3,4}]\label{ulb}~~\\
	The solution $X(t)=(S(t),E(t),I(t),Q(t),R(t))$ of model \eqref{sys1} is said to be stochastically ultimately bounded, if for any $\epsilon\in (0,1),$ there is a positive constant $\chi=\chi(\epsilon)$, such that for any initial value $X(0)=X_0\in\mathbb{R}_+^5$  the solution $X(t)$ has the property
	\begin{eqnarray*}
		\limsup\limits_{t \rightarrow +\infty}\mathbb{P}(\vert\vert X(t,X_0)\vert\vert > \chi )< \epsilon.
	\end{eqnarray*}
\end{definition}

\begin{remark}\label{rr}
	In \cite{7,8}, the authors gave a new definition for stochastically ultimate boundedness, according to which the system \eqref{sys1} is stochastically bounded in probability, if for any $\epsilon\in (0,1),$ there is some $\chi=\chi(\epsilon)>0$ for which the following inequality:
	\begin{eqnarray}\label{gou}
		\limsup\limits_{t \rightarrow +\infty}\mathbb{P}(\vert\vert X(t,X_0)\vert\vert \leq \chi )\geq 1- \epsilon.
	\end{eqnarray}
is satisfied for any initial value $X(0)=X_0\in\mathbb{R}_+^5$. We mention that the inequality \eqref{gou} can equivalently be rewritten as:
	\begin{eqnarray*}\label{gou}
	\liminf\limits_{t \rightarrow +\infty}\mathbb{P}(\vert\vert X(t,X_0)\vert\vert > \chi )\leq \epsilon.
\end{eqnarray*}
\end{remark}

\begin{definition}[Stochastic persistence\label{zlb}  \cite{8}]~~\\
	The model \eqref{sys1} is said to be stochastically persistent if for any $\epsilon\in (0,1),$ there exists a positive constants $\kappa=\kappa(\epsilon)$ such that for any initial value $X(0)=X_0\in\mathbb{R}_+^5$ the solution $X(t)$ has the property
	\begin{eqnarray*}
\liminf\limits_{t \rightarrow +\infty}\mathbb{P}(\vert\vert X(t,X_0)\vert\vert\geqslant \kappa )\geqslant 1-\epsilon.
	\end{eqnarray*}
\end{definition}

\begin{definition}[Stochastic permanence \cite{9}]~~\\
The system \eqref{sys1}is called stochastically permanent, if it is both stochastically ultimate bounded and persistent.
\end{definition}

\begin{theorem}
	The model \eqref{sys1} is stochastically  ultimate bounded and stochastically permanent for any initial value $X_0\in\mathbb{R}^5_+$.
\end{theorem}
\begin{proof}
Making use of It\^o's formula to $e^{\mu t}V(X(t))$ and inequality \eqref{izne}, we get
	\begin{eqnarray}\label{mimou}
		d (e^{\mu t}V(X(t)))&=& \mu e^{\mu t} V(X(t))dt+e^{\mu t}d V(X(t))\nonumber\\\nonumber
		&\leq & \mu e^{\mu t}V(X(t))dt+e^{\mu t}\bigg[(\varpi-\mu V(X(t)))dt+\bigg(1-\dfrac{1}{N^2(t)}\bigg)[\sigma_1 S(t) dB_1(t)\\&& \nonumber+\sigma_2 E(t) dB_2(t)+\sigma_3 I(t) dB_3(t)+\sigma_4 Q(t) dB_4(t)+\sigma_5 R(t) dB_5]\bigg]\\\nonumber
		&\leq & \varpi e^{\mu t}dt+e^{\mu t}\bigg(1-\dfrac{1}{N^2(t)}\bigg)[\sigma_1 S(t) dB_1(t)+\sigma_2 E(t) dB_2(t)+\sigma_3 I(t) dB_3(t)\nonumber\\&&+\sigma_4 Q(t) dB_4(t)+\sigma_5 R(t) dB_5(t)]
	\end{eqnarray} 
By integrating over the interval \( [0, t \wedge \tau_k] \), where \(\wedge\) denotes the minimum operator (\(\min\)), and then taking the expectation on both sides of \eqref{mimou}, we obtain for all $t \geq 0$ and $k\geq k_0$
\begin{eqnarray*}
\mathbb{E}\bigg[e^{\mu (t\wedge \tau_k)}V(X(t\wedge \tau_k))\bigg]\leq V(X(0))+\mathbb{E}\bigg[\displaystyle\int_{0}^{t\wedge \tau_k} \varpi e^{\mu s}ds\bigg]\leq V(X(0))+\dfrac{\varpi}{\mu}(e^{\mu t}-1).
\end{eqnarray*}	
Letting $k\rightarrow\infty$ yields
\begin{eqnarray*}
	\mathbb{E}\bigg[V(X(t))\bigg]\leq V(X(0))e^{-\mu t} +\dfrac{\varpi}{\mu}(1-e^{-\mu t})\leq  V(X(0))e^{-\mu t} +\dfrac{\varpi}{\mu}.
\end{eqnarray*}
Let $\epsilon>0$, and choose $\rho_{\epsilon} = \dfrac{\varpi}{\epsilon \mu}$. By Markov's inequality, we have
\begin{eqnarray*}
\mathbb{P}\bigg(V(X(t)) >\rho_{\epsilon}\bigg)\leq \dfrac{1}{\rho_{\epsilon}}\mathbb{E}[V(X(t))]\leq \dfrac{1}{\rho_{\epsilon}}\bigg(V(X(0))e^{-\mu t}\bigg)+\epsilon.
\end{eqnarray*}
This implies
\begin{eqnarray*}
	\mathbb{P}\bigg(N(t)+\dfrac{1}{N(t)}\leq\rho_{\epsilon}\bigg) \geq 1-\epsilon-\dfrac{1}{\rho_{\epsilon}}\bigg(V(X(0))e^{-\mu t}\bigg).
\end{eqnarray*}
Therefore
\begin{eqnarray*}
\mathbb{P}\bigg(\dfrac{1}{\rho_{\epsilon}}\leq N(t)\leq \rho_{\epsilon} \bigg)	\geq\mathbb{P}\bigg(N(t)+\dfrac{1}{N(t)}\leq\rho_{\epsilon}\bigg) \geq 1-\epsilon-\dfrac{1}{\rho_{\epsilon}}\bigg(V(X(0))e^{-\mu t}\bigg).
\end{eqnarray*}
By noting that
$$N^2\leq 5\vert\vert X\vert\vert^2\leq 5 N^2, $$
we obtain
$$1-\epsilon-\dfrac{1}{\rho_{\epsilon}}\bigg(V(X(0))e^{-\mu t}\bigg)\leq \mathbb{P}\bigg(\dfrac{1}{\rho_{\epsilon}}\leq N(t)\leq \rho_{\epsilon} \bigg)\leq \mathbb{P}\bigg(\dfrac{1}{\sqrt{5}\rho_{\epsilon}}\leq\vert\vert X(t)\vert\vert\leq\rho_{\epsilon} \bigg).$$
Hence
$$\liminf\limits_{t\rightarrow\infty}\mathbb{P}\bigg(\dfrac{1}{\sqrt{5}\rho_{\epsilon}}\leq\vert\vert X(t)\vert\vert\leq\rho_{\epsilon} \bigg)\geq 1-\epsilon.$$
According to Definition \ref{ulb} and Definition \ref{zlb}, the model \eqref{sys1} is stochastically ultimately bounded and permanent. The proof is complete.
\end{proof}

\section{Stochastic extinction }
In this subsection, we will try our best to find a condition for the extinction of the disease expressed in terms of system parameters and intensities of noises.\\\
Before stating the result to be proved, we must firstly give the following useful lemma that was stated and proved as Lemma 2.1 in \cite{10}.
\begin{lemma}\label{ab}
For any initial value $X_0 \in\mathbb{R}^5_+$, the solution $X(t) = (S(t), E(t), I(t), Q(t), R(t))$ of system \eqref{sys1} verifies the following properties:
\begin{eqnarray*}
(a)&~~&\lim\limits_{t\rightarrow\infty}\dfrac{X_k(t)}{t}=0~~\mbox{a.s}~~\forall~k\in\{1,2,...,5\},\\
(b)&~~&\mbox{Moreover, if } \mu>\frac{1}{2}(\sigma_1^2\vee\sigma_2^2\vee\sigma_3^2\vee\sigma_4^2\vee\sigma_5^2),~~\mbox{then}
\lim\limits_{t\rightarrow\infty}\dfrac{\displaystyle\int_{0}^{t}X_k(s)dB_k(s)}{t}=0~~\mbox{a.s}~~\forall~k\in\{1,2,...,5\}.
\end{eqnarray*}	
\end{lemma}
\begin{proof}
The proof of this Lemma \ref{ab} is similar in spirit to that of Lemmas 2.1 and 2.2 of \cite{10} and therefore it is omitted here. 
\end{proof}
\begin{definition}[Stochastic extinction \cite{9}]
For system \eqref{sys1}, the infected individuals $E(t)$ and $I(t)$ are said to be stochastically extinct, or extinctive, if $\lim\limits_{t\rightarrow \infty} (E(t) + I(t))  = 0$ almost surely.
\end{definition}

\begin{theorem}\label{thmm}
Let $(S(t), E(t), I(t), Q(t), R(t))$ be the solution of system \eqref{sys1} with any initial value $(S(0), E(0), I(0), Q(0), R(0))\in\mathbb{R}_+^5$.\\
If $\mu>\frac{1}{2}(\sigma_1^2\vee\sigma_2^2\vee\sigma_3^2\vee\sigma_4^2\vee\sigma_5^2)$ and $\sigma_1^2\wedge\sigma_3^2>4(\beta h'(0)\tilde{S}_0-\mu)$, with $\tilde{S}_0=\dfrac{b}{\eta+\mu+d_3},$ then
\begin{eqnarray*}
	\limsup\limits_{t\rightarrow\infty}\dfrac{\ln(E(t)+I(t))}{t}&\leq& \beta h'(0)\tilde{S}_0-\mu-\dfrac{\sigma_2^2\wedge\sigma_3^2}{4}<0~~\mbox{a.s,}
\end{eqnarray*}
which means that the disease will die out exponentially with probability one.
\end{theorem}

\begin{proof}
From It\^o's formula
\begin{eqnarray*}
	d\ln(E+I)&=&\bigg[\dfrac{1}{E+I}\bigg(\beta Sh(I)-d_2 E-(\nu+\gamma+d_1)I\bigg)-\mu-\dfrac{\sigma_2^2E^2+\sigma_3^2 I^2}{2(E+I)^2}\bigg]dt\\&&+ \sigma_2\dfrac{E}{E+I}dB_2+\sigma_3\dfrac{I}{E+I}dB_3.
\end{eqnarray*}
Thus
\begin{eqnarray*}
	d\ln(E+I)&\leq& \bigg[\beta h'(0)S-\mu-\dfrac{\sigma_2^2\wedge\sigma_2^3}{2}\dfrac{E^2+I^2}{(E+I)^2}\bigg]dt+\sigma_2\dfrac{E}{E+I}dB_2+\sigma_3\dfrac{I}{E+I}dB_3.
\end{eqnarray*}
Using the inequality $(a+b)^2\leq 2a^2+2b^2$ for all $a,b\in\mathbb{R},$ we get
\begin{eqnarray}\label{stoex}
	d\ln(E+I)&\leq & \bigg[\beta h'(0)S-\mu-\dfrac{\sigma_2^2\wedge\sigma_2^3}{4}\bigg]dt+\sigma_2\dfrac{E}{E+I}dB_2+\sigma_3\dfrac{I}{E+I}dB_3.
\end{eqnarray}
Integrating \eqref{stoex} from $0$ to $t$, and then dividing by $t$ on both sides, we obtain
\begin{eqnarray}\label{rtt}
\dfrac{\ln(E(t)+I(t))}{t}&\leq&\dfrac{\ln(E(0)+I(0))}{t} +\beta h'(0)\langle S(t)\rangle-\mu-\dfrac{\sigma_2^2\wedge\sigma_2^3}{4}\nonumber\\&&+\dfrac{\sigma_2}{t}\displaystyle\int_{0}^{t}\dfrac{E(s)}{E(s)+I(s)}dB_2(s)+\dfrac{\sigma_3}{t}\displaystyle\int_{0}^{t}\dfrac{I(s)}{E(s)+I(s)}dB_3(s).
\end{eqnarray}
On the other hand, the first equation of \eqref{sys1} gives
\begin{eqnarray*}
S(t)-S(0)&=& bt-\beta\displaystyle\int_{0}^{t} S(s)h(I(s)) ds-(\eta+\mu+d_3)\displaystyle\int_{0}^{t}S(s)ds+\sigma_1\displaystyle\int_{0}^{t}S(s)dB_1(s)\\
&\leq & bt-(\eta+\mu+d_3)\displaystyle\int_{0}^{t}S(s)ds+\sigma_1\displaystyle\int_{0}^{t}S(s)dB_1(s).
\end{eqnarray*}
Therefore
\begin{eqnarray}\label{poly}
\langle S(t)\rangle&\leq&\dfrac{1}{\eta+\mu+d_3}\bigg[b+\dfrac{S(0)-S(t)}{t}+\dfrac{\sigma_1}{t}\displaystyle\int_{0}^{t}S(s)dB_1(s)\bigg].
\end{eqnarray}
Since $\mu>\frac{1}{2}(\sigma_1^2\vee\sigma_2^2\vee\sigma_3^2\vee\sigma_4^2\vee\sigma_5^2)$, we can conclude by virtue of Lemma \ref{ab} and inequality \eqref{poly} that

\begin{eqnarray}\label{ttr}
\lim\limits_{t\rightarrow\infty}\langle S(t)\rangle &\leq&\dfrac{b}{\eta+\mu+d_3}=S_0.
\end{eqnarray}
From \eqref{rtt} and  \eqref{ttr}, we get
\begin{eqnarray*}
\limsup\limits_{t\rightarrow\infty}\dfrac{\ln(E(t)+I(t))}{t}&\leq& \beta h'(0)\tilde{S}_0-\mu-\dfrac{\sigma_2^2\wedge\sigma_3^2}{4}<0~~\mbox{a.s,}
\end{eqnarray*}
which completes the proof.
\end{proof}

\begin{corollary}
Under the same notations and hypotheses as in Theorem \ref{thmm}, we have\\
$$\lim\limits_{t\rightarrow\infty}\langle S(t) \rangle= \dfrac{b}{\eta+\mu+d_3} ~~\mbox{a.s.,}$$ $$ \lim\limits_{t\rightarrow\infty}\langle Q(t) \rangle= \dfrac{bd_3}{(\mu+\tau)(\eta+\mu+d_3)} ~~\mbox{a.s.,}$$ $$\lim\limits_{t\rightarrow\infty}\langle R(t) \rangle= \dfrac{b[\eta(\mu+\tau)+\tau d_3]}{\mu(\eta+\mu+d_3)(\mu+\tau)} ~~\mbox{a.s.} $$
\end{corollary}
\begin{proof}
We have
\begin{eqnarray}\label{co}
d(S(t)+E(t))=[b-(\eta+\mu+d_3)S(t)-(\lambda+\mu+d_2)E(t)]dt+\sigma_1 S(t)dB_1+\sigma_2 E(t)dB_2.
\end{eqnarray}
Integrating \eqref{co} from $0$ to $t$, and then dividing by $t$ on both sides yields

\begin{eqnarray*}
\dfrac{S(t)+E(t)}{t}-\dfrac{S(0)+E(0)}{t}=b-(\eta+\mu+d_3)\langle S(t) \rangle-(\lambda+\mu+d_2)\langle E(t) \rangle+\dfrac{\sigma_1 }{t} \displaystyle\int_{0}^tS(t)dB_1+\dfrac{\sigma_2}{t}\displaystyle\int_{0}^t E(t)dB_2.
\end{eqnarray*}
Hence
\begin{eqnarray}\label{op}
(\eta+\mu+d_3)\langle S(t) \rangle=b+\dfrac{S(0)+E(0)}{t}-\dfrac{S(t)+E(t)}{t}-(\lambda+\mu+d_2)\langle E(t) \rangle+\dfrac{\sigma_1 }{t} \displaystyle\int_{0}^tS(t)dB_1+\dfrac{\sigma_2}{t}\displaystyle\int_{0}^t E(t)dB_2.
\end{eqnarray}
Letting $t$ go to infinity on both sides of \eqref{op}, and using  Lemma \ref{ab}, we obtain
\begin{eqnarray*}
\lim\limits_{t\rightarrow\infty}(\eta+\mu+d_3)\langle S(t) \rangle=b-(\lambda+\mu+d_2)\lim\limits_{t\rightarrow\infty}\langle E(t)\rangle~~\mbox{a.s.}.
\end{eqnarray*}
In virtue of Theorem \eqref{thmm}, we have
$$\lim\limits_{t\rightarrow\infty} E(t)=0~~\mbox{a.s.},$$
which implies by the continuous version of Cesàro’s theorem \cite{15} that
$$\lim\limits_{t\rightarrow\infty}\langle E(t)\rangle=0~~\mbox{a.s.}.$$
Therefore
$$\lim\limits_{t\rightarrow\infty}\langle S(t)\rangle =\dfrac{b}{\eta+\mu+d_3}~~\mbox{a.s.}.$$
At the same time, we have
\begin{eqnarray*}
dQ(t)=[d_3 S(t)+d_2 E(t)+d_1 I(t)-(\mu+\tau)Q(t)]dt+\sigma_4 Q(t) dB_4(t).
\end{eqnarray*}
Then 
\begin{eqnarray*}
\dfrac{Q(t)-Q(0)}{t}=d_3 \langle S(t) \rangle+d_2\langle E(t) \rangle+d_1 \langle I(t) \rangle-(\mu+\tau)\langle Q(t) \rangle+\dfrac{\sigma_4}{t}\displaystyle\int_{0}^t Q(t)dB_4.
\end{eqnarray*}
Hence
\begin{eqnarray*}
\langle Q(t) \rangle=\dfrac{d_3}{\mu+\tau}\langle S(t) \rangle+\dfrac{d_2}{\mu+\tau}\langle E(t) \rangle+\dfrac{d_1}{\mu+\tau}\langle I(t) \rangle+\dfrac{Q(0)-Q(t)}{t(\mu+\tau)}+\dfrac{\sigma_4}{t(\mu+\tau)}\displaystyle\int_{0}^t Q(t)dB_4.
\end{eqnarray*}
By passage to the limit similar to the above, we get
$$\lim\limits_{t\rightarrow\infty}\langle Q(t)\rangle =\dfrac{bd_3}{(\mu+\tau)(\eta+\mu+d_3)}~~\mbox{a.s.}.$$
The same reasoning 
$$\lim\limits_{t\rightarrow\infty}\langle R(t)\rangle =\dfrac{b[\eta(\mu+\tau)+\tau d_3]}{\mu(\eta+\mu+d_3)(\mu+\tau)}~~\mbox{a.s.}.$$
\end{proof}

\section{$V-$geometric ergodicity}

\begin{definition}[ Stationary distribution\cite{11}]
 Denote $\mathbb{P}_{\delta}$ the corresponding probability distribution of an initial distribution $\delta$ , which describes the initial state of  at $t = 0$. Suppose that the distribution of $X(t)$ with initial distribution $\delta $converges in some sense to a distribution $\pi=\pi_{\delta}$ (a priori $\pi$ may depend on the initial distribution $\delta$ ), i.e.,
 \begin{eqnarray*}
 	\lim\limits_{t \rightarrow +\infty}\mathbb{P}_{\delta}\{X(t)\in \mathcal{A}\}=\pi(\mathcal{A}),~~\mbox{for all measurable } \mathcal{A}.
 \end{eqnarray*}
Then, we say that  has a stationary distribution $\pi(.)$.
\end{definition}

\begin{lemma}\label{kko}\cite{12,13} Assume that the following assumptions hold: \\
\textbf{(P$1$)} (Minorization condition) For a compact set $D_1 \subset \mathbb{R}_{+}^5$, there exist $T, \alpha>0$ and a probability measure $v$ on $\mathbb{R}_{+}^5$ with $v\left(U_1\right.$ )
$>0$ such that
$$
P_T\left(X_0,\mathcal{A} \right) \geq \alpha v(A), \forall X_0 \in U_1, \forall \mathcal{A} \in \mathcal{B}\left(\mathbb{R}_{+}^5\right)
$$
\textbf{(P$2$)} (Lyapunov condition) There is a function $V: \mathbb{R}_{+}^5 \rightarrow[1, \infty)$ with $\lim _{|X(t)| \rightarrow \infty} V(X)=\infty$ and real numbers $\beta_1, \beta_2 \in(0, \infty)$ such that
$$
\mathcal{L}V(X) \leq \beta_1-\beta_2 V(X)
$$
Then the Markov process $X(t)$ is $V$-geometrically ergodic: there exists a unique stationary distribution $\pi$ such that, for some constants $C,\Lambda>0$
$$
|\mathbb{E} g(X(t))-\pi(g)| \leq C V\left(X_0\right) e^{-\Lambda t}, \forall X(0)=X_0 \in \mathbb{R}_{+}^5
$$
for all measurable function $g \in \mathcal{G}:=\left\{\right.$ measurable $g: \mathbb{R}_{+}^5 \rightarrow \mathbb{R}^5$ with $\left.|g(X)| \leq V(X)\right\}$.
\end{lemma}

The following result shows the existence of $V-$geometric ergodicity of the Markov process $X(t) = (S(t), E(t),I(t),Q(t), R(t))$ of model \eqref{sys1}.
\begin{theorem}
	Markov process $X(t)$ of model \eqref{sys1} with initial value $X_0 \in\mathbb{R}_+^5$ is $V-$geometrically ergodic.
\end{theorem}
\begin{proof}
	Summing up the five equations in \eqref{sys1} and denoting $N(t) = S(t) + E(t)+ I(t) + Q(t)+R(t)$, we obtain for all $t\geq 0$
	\begin{eqnarray*}
		dN(t)&=&(b-\mu N(t)-\nu I(t))dt+\sigma_1 S(t)dB_1(t)+\sigma_2 E(t)dB_2(t)+\sigma_3 I(t)dB_3(t)+\sigma_4 Q(t)dB_4(t)\\&&+\sigma_5 R(t)dB_5(t).
	\end{eqnarray*}
	Define a $\mathcal{C}^2-$function $V(X(t))=N+\dfrac{1}{N}.$\\
	For $X(t)\in \mathbb{R}^5_+$, it follows that $V(X(t))\rightarrow \infty$ as $\vert X(t)\vert \rightarrow \infty$. By It\^o's formula, we have
	\begin{eqnarray*}
		dV(X(t))&=& \mathcal{L}V(X(t))dt+\bigg(1-\dfrac{1}{N^2(t)}\bigg)[\sigma_1 S(t)dB_1(t)+\sigma_2 E(t)dB_2(t)+\sigma_3 I(t)dB_3(t)\\&&+\sigma_4 Q(t)dB_4(t)+\sigma_5 R(t)dB_5(t)],
	\end{eqnarray*}
	where $\mathcal{L}V(X)$ is given by
	\begin{eqnarray}\label{izne}
		\mathcal{L}V(X)&=&(b-\mu N-\nu I)-\dfrac{b-\mu N-\nu I}{N^2}+\dfrac{1}{N^3}(\sigma_1^2S^2+\sigma_2^2E^2+\sigma_3^2I^2+\sigma_4^2Q^2+\sigma_5^2R^2)\nonumber\\ \nonumber
		&\leq & (b-\mu N)-\dfrac{b}{N^2}+\dfrac{\mu+\nu}{N}+\dfrac{1}{N}\sum\limits_{i=1}^{5}\sigma_i^2\\ \nonumber
		&\leq &  -\mu(N+\dfrac{1}{N})+b-\dfrac{b}{N^2}+\dfrac{1}{N}\bigg(2\mu+\nu+\sum\limits_{i=1}^{5}\sigma_i^2\bigg)\\ \nonumber
		&\leq & -\mu (N+\dfrac{1}{N})-\bigg(\dfrac{\sqrt{b}}{N}-\dfrac{1}{2\sqrt{b}}\bigg(2\mu+\nu+\sum\limits_{i=1}^{5}\sigma_i^2\bigg)\bigg)^2+\dfrac{1}{4b}\bigg(2\mu+\nu+\sum\limits_{i=1}^{5}\sigma_i^2\bigg)^2+b\\
		&\leq & \varpi-\mu V(X), 
	\end{eqnarray}	
	where \\$\varpi=\dfrac{1}{4b}\bigg(2\mu+\nu+\sum\limits_{i=1}^{5}\sigma_i^2\bigg)^2+b$.
	\\\\
Hence, in Lemma \ref{kko} condition \textbf{(P2)} holds. Since model \eqref{sys1} is uniformly elliptic, according to Proposition 11.1 in \cite{14} we can find a function $g:\mathbb{R}_+\times\mathbb{R}_+^5\times\mathbb{R}_+^5\rightarrow (0,\infty)$ such that $g$ is jointly continuous and  $g_t(X_0, Y)$ is strictly positive for all $(t, X_0, Y)$, such that for all measure sets $\mathcal{A}$ 
\begin{eqnarray*}
P_t(X_0,\mathcal{A})=\displaystyle\int_{\mathcal{A}}g_t(X_0,Y) dY.	
\end{eqnarray*}
It follows that for any $\omega > 0$, there exists a positive constant $a = a(\omega,t) > 0$ so that $\inf\{P_t(X_0,Y) : X_0,~Y\in \mathbb{R}^5_+,~\lvert X_0\rvert,\lvert Y\rvert \leq \omega\} \geq a.$ The hypothesis \textbf{(P1)} implies the argument, as it known that for any measurable set $A$ we have
\begin{eqnarray*}
P_t(X_0,\mathcal{A})=\displaystyle\int_{\mathcal{A}}g_t(X_0,Y) dY\geq a~\mbox{Leb}(\mathcal{A}\cap \mathcal{B}_{\omega}(0))=a\mbox{Leb}(\mathcal{B}_{\omega}(0))v(\mathcal{A}),
\end{eqnarray*}
where Leb is Lebesgue measure and $v(\mathcal{A}) = Leb (\mathcal{A} \cup(\mathcal{B}_{\omega}(0))/Leb (\mathcal{B}_{\omega}(0))$. Hence we obtain the required result.
\end{proof}

\section{Numerical Simulations}\label{sec:numerics}

In this section we present detailed numerical illustrations of the dynamics of a generalized stochastic SEIQR epidemic model, emphasizing the effectiveness of numerical solutions for stochastic differential equations (SDEs). Through a series of simulations, we demonstrate the capability of our approach to capture the intricate behaviors of epidemic models influenced by stochastic processes, particularly Brownian motions, which serve to represent the inherent randomness associated with the spread of infectious diseases.

To address the stochastic nature of the SEIQR model, we use the truncated Milstein method as our primary numerical discretization technique for solving the governing SDEs. We consider several illustrative examples to showcase the robustness of our methodology in modeling epidemic dynamics under uncertainty.

\subsection{Time Discretization}
Roughly speaking, we are given a diffusion process:
\begin{equation}\label{eq:diff_proc}
dX_t = \alpha(X_t)dt + \beta(X_t)dB_t,
\end{equation}
where $X_0=x_0\in\mathbb{R}^d$ is given, $\alpha:\mathbb{R}^d\rightarrow\mathbb{R}^d$, $\beta:\mathbb{R}^d\rightarrow\mathbb{R}^{d\times d}$ and $\{B_t\}_{t\geq 0}$ is a
standard $d-$dimensional Brownian motion.

Typically one must time discretize \eqref{eq:diff_proc} and we consider a time discretization at equally spaced times, separated by $\Delta_l=2^{-l}$. To continue with our exposition, we define the $d-$vector, $F:\mathbb{R}^{2d}\times\mathbb{R}^+\rightarrow\mathbb{R}^d$, 
$F_{\Delta}(x,z)=(F_{\Delta,1}(x,z),\dots,F_{\Delta,d}(x,z))^{\top}$ where for $i\in\{1,\dots,d\}$
\begin{eqnarray*}
F_{\Delta,i}(x,z) & = &  \sum_{(j,k)\in\{1,\dots,d\}^2} f_{ijk}(x)(z_jz_k-\Delta) \\
f_{ijk}(x) & = & \frac{1}{2}\sum_{m\in\{1,\dots,d\}} \beta_{mk}(x)\frac{\partial \beta_{ij}(x)}{\partial x_m}.
\end{eqnarray*}
Here, $\mathcal{N}_d(\mu,\Sigma)$ denotes the $d$-dimensional Gaussian distribution with mean vector $\mu$ and covariance matrix $\Sigma$; if $d=1$, we drop the subscript $d$. $I_d$ represents the $d\times d$ identity matrix.
A single-level implementation of the truncated Milstein scheme \cite{Ajay1, Ajay2, Milstein}, a numerical method that is a key focus of this article, is detailed in Algorithm\ref{alg:milstein_sl}.

\begin{algorithm}[H]
\begin{enumerate}
\item{Input level $l$ and starting point $x_0^l$.}
\item{Generate $Z_k\stackrel{\textrm{i.i.d.}}{\sim}\mathcal{N}_d(0,\Delta_l I_d)$, $k\in\{1,2,\dots,\Delta_l^{-1}\}$.}
\item{Generate level $l$: for $k\in\{0,1,\dots,\Delta_l^{-1}-1\}$ with $X_0^l=x_0^l$
$$
X_{(k+1)\Delta_l}^l = X_{k\Delta_l}^l + \alpha(X_{k\Delta_l}^l)\Delta_l + \beta(X_{k\Delta_l}^l)Z_{k+1} + F_{\Delta_l}(X_{k\Delta_l}^l,Z_{k+1}).
$$
}
\item{Output $X_1^l$.}
\end{enumerate}
\caption{Truncated Milstein Scheme on $[0,1]$.}
\label{alg:milstein_sl}
\end{algorithm}

\subsection{Set-Up}

We consider the SEIQR model, where the state vector \( X_t = (S_t, E_t, I_t, Q_t, R_t) \in \mathbb{R}_{+}^5 \) represents the proportions of susceptible, exposed, infectious, quarantined, and recovered individuals at time \( t \in [0,T] \), respectively. The model is designed to capture the dynamics of disease transmission and progression within a population over a specified time period.

The transmission function \( h(x) \) is defined as a fractional positive increasing function on the domain \( [0, \infty) \), formulated as \( h(x) = \frac{x}{1 + 0.01x} \). This function ensures that as the value of \( x \) increases, \( h(x) \) asymptotically approaches a linear growth, while still remaining bounded. Notably, \( h(0) = 0 \), which implies that in the absence of an infectious agent, the transmission rate remains zero.

For the initial conditions, we assume that the proportions of the population in each compartment at \( t = 0 \) are equal, with \( S_0 = E_0 = I_0 = Q_0 = R_0 = 0.25 \). This choice of initial conditions ensures a balanced starting point for the simulation, allowing us to observe the natural evolution of the system under the specified model parameters.

The SEIQR system parameters are chosen as follows: \( \mu = 0.55 \), \( b = 1 \), \( \eta = 0.01 \), \( d_3 = 0.01 \), \( \beta = 0.1 \), \( \lambda = 0.1 \), \( d_2 = 0.03 \), \( \nu = 0.01 \), \( \gamma = 0.05 \), \( d_1 = 0.1 \), and \( \tau = 0.05 \). 

For numerical simulations, we employ a discretization approach with a level \( 9 \) refinement, corresponding to a time step size of \( \Delta_l = 2^{-9} = 0.002 \). This fine discretization ensures that the model's dynamics are captured with high accuracy, allowing for a detailed analysis of the temporal evolution of the system.

\subsection{Numerical Results}

Our analysis begins with simulating the stochastic SEIQR system using MATLAB software. We set the diffusion coefficients to \(\sigma_1 = \sigma_4 = \sigma_5 = 0.25\) and \(\sigma_2 = \sigma_3 = 1\), which introduces different levels of uncertainty into each compartment of the model. By including these stochastic elements, we aim to see how random fluctuations might affect the long-term behavior of the system.

To ensure our model is working correctly, we first verify that the parameters meet key theoretical requirements. For example, we confirm that the drift parameter \(\mu = 0.55\) is larger than half the maximum of the squared diffusion coefficients, which is \(0.5\). We also check that \(\sigma_1^2 \wedge \sigma_3^2 = 0.0625\) is greater than $4(\beta h'(0) \tilde{S}_0 - \mu) = -1.4982,$ where \(\tilde{S}_0 = 1.7544\) represents the average number of susceptible individuals in the long run. Our analysis of the system’s behavior over time reveals some interesting results. As time goes on, the average number of susceptible individuals stabilizes at:

\[
\lim_{t \rightarrow \infty} \langle S(t) \rangle = \frac{b}{\eta + \mu + d_3} = 1.7544.
\]

Similarly, the long-term averages for the quarantined and recovered populations are:

\[
\lim_{t \rightarrow \infty} \langle Q(t) \rangle = \frac{b d_3}{(\mu + \tau)(\eta + \mu + d_3)} = 0.0292
\]

and

\[
\lim_{t \rightarrow \infty} \langle R(t) \rangle = \frac{b[\eta(\mu+\tau) + \tau d_3]}{\mu(\eta + \mu + d_3)(\mu + \tau)} = 0.0346.
\]

Figure \ref{fig: Extinction} displays the trajectories of the SEIQR model over the time interval from $0$ to $100$, highlighting the model's behavior under the specified synthetic parameters. Each panel represents a different compartment—susceptible, exposed, infectious, quarantined, and recovered—and illustrates how these populations evolve over time. The figure also captures the extinction behavior, revealing how certain compartments approach extinction, which provides crucial insights into the model's long-term dynamics and the impact of the parameters on the system. As indicated in Theorem 5.1, this result confirms that the conditions outlined in the theorem are satisfied in this case. However, when the conditions outlined in the theory are not met, the system exhibits stochastic dynamic permanence. This behavior is clearly illustrated in Figure \ref{fig: Persistence}, where the populations continue to fluctuate without approaching extinction. In contrast to extinction, where certain compartments eventually disappear, dynamic permanence implies that the epidemic persists indefinitely, with fluctuations in the population levels of the compartments. This outcome aligns with the predictions made by the permanence theorem discussed earlier, further validating the model's ability to capture long-term epidemic behavior under varying conditions.


\begin{figure}[H]
\centering
\subfloat[Susceptible Population \(S(t)\)]{\includegraphics[width=0.45\textwidth]{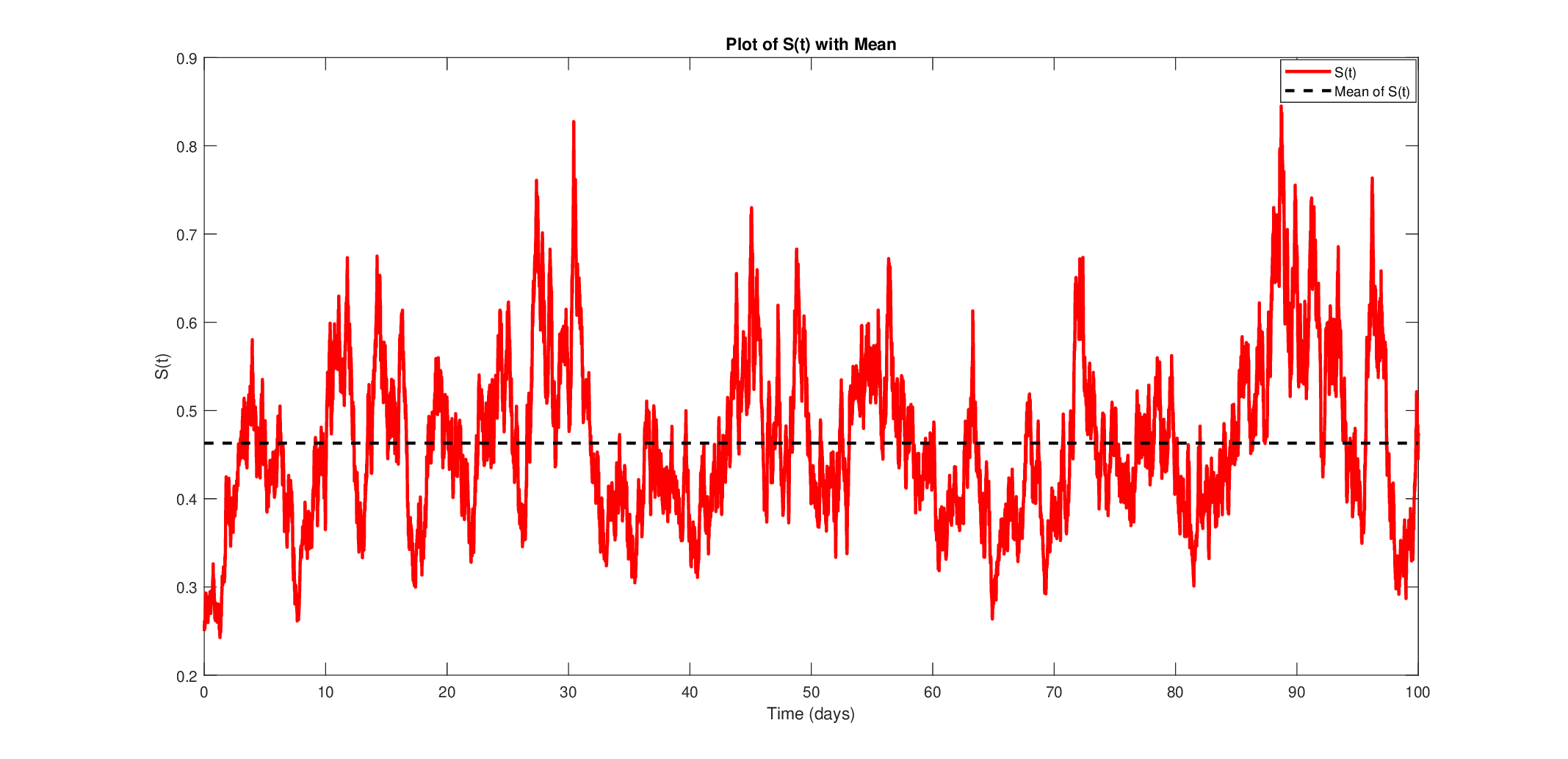}}\hspace{0.5cm}
\subfloat[Exposed Population \(E(t)\)]{\includegraphics[width=0.45\textwidth]{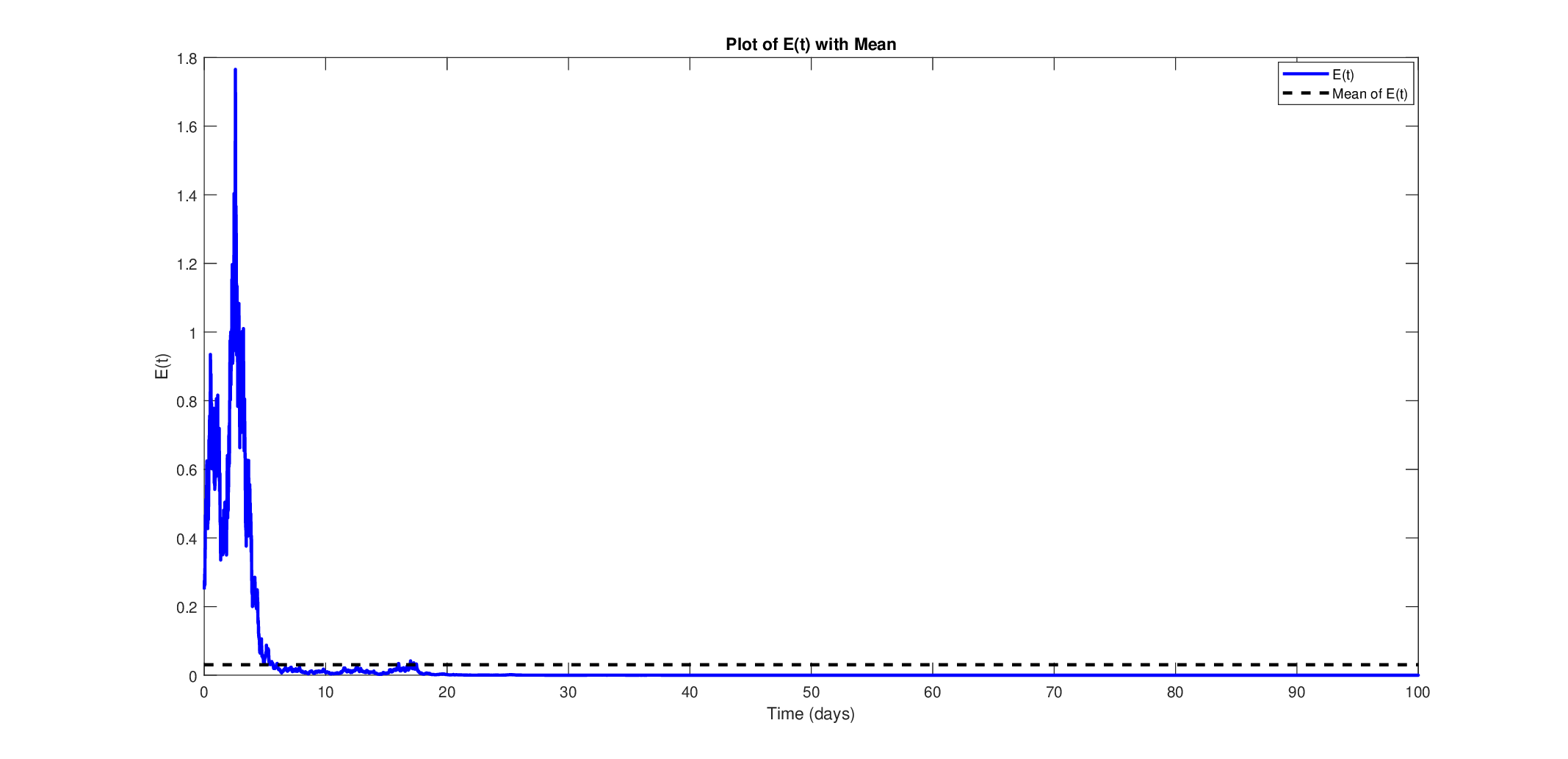}}\\
\subfloat[Infectious Population \(I(t)\)]{\includegraphics[width=0.45\textwidth]{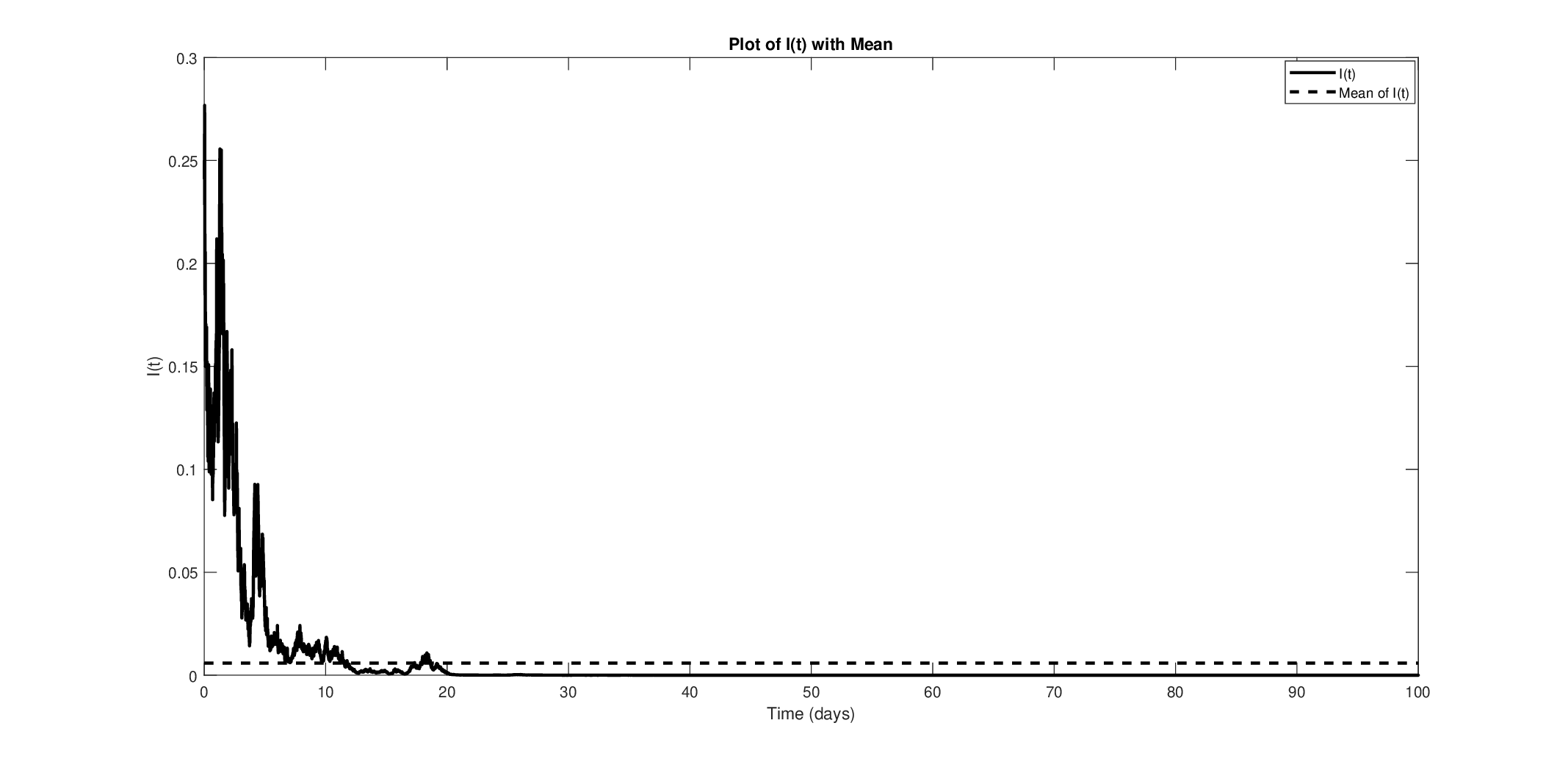}}\hspace{0.5cm}
\subfloat[Quarantined Population \(Q(t)\)]{\includegraphics[width=0.45\textwidth]{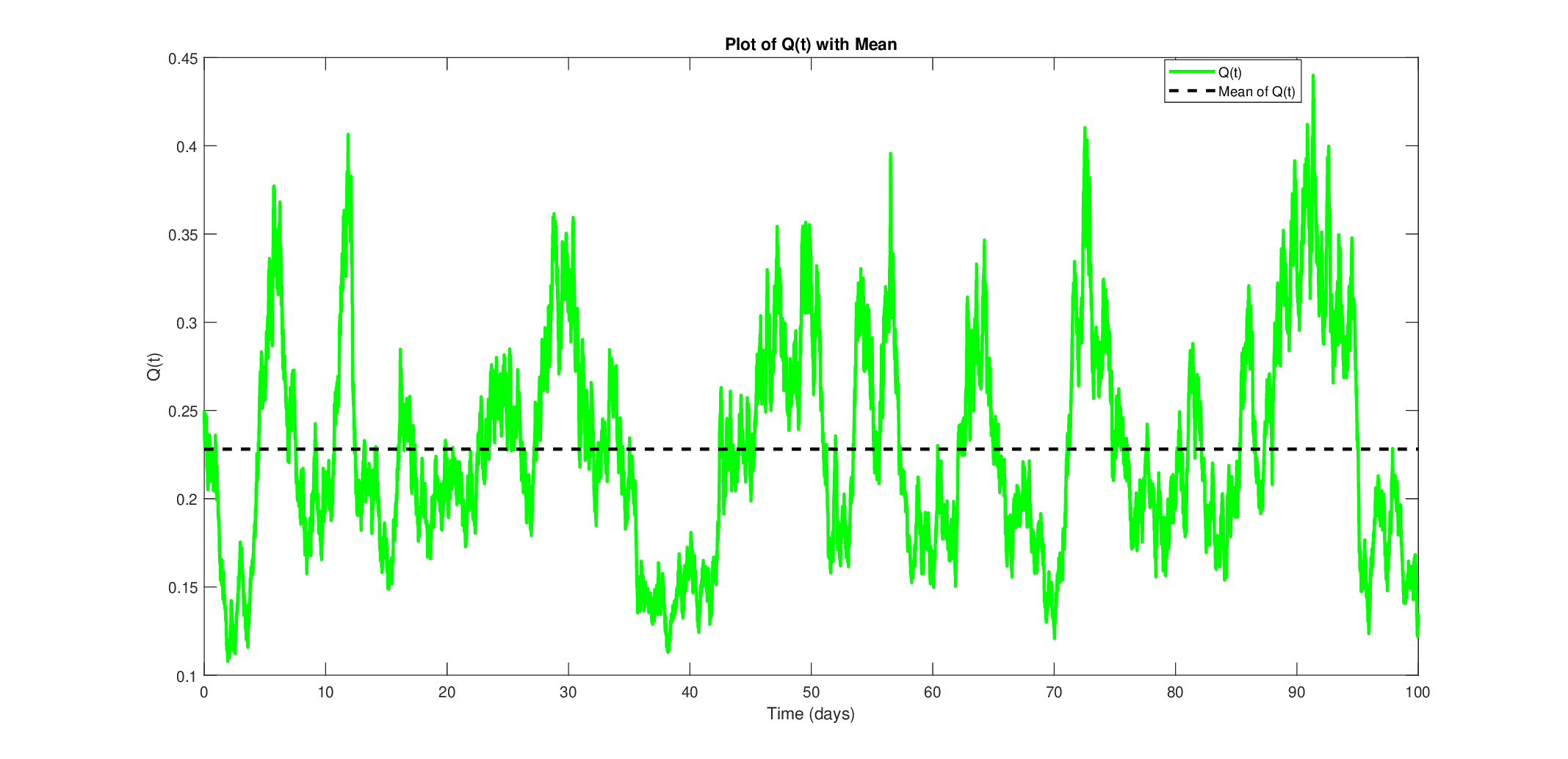}}\\
\subfloat[Recovered Population \(R(t)\)]{\includegraphics[width=0.6\textwidth]{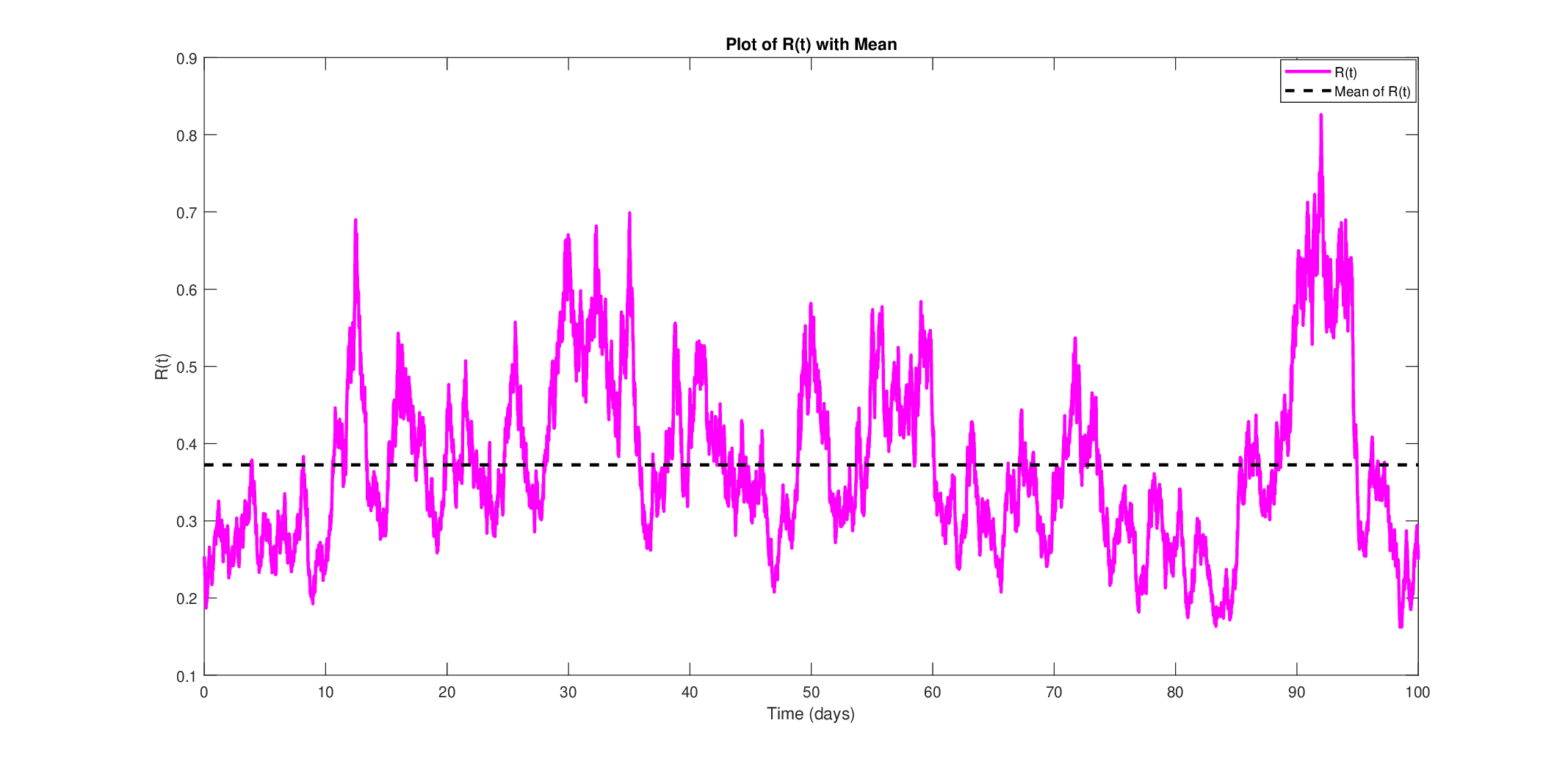}}

\caption{Trajectories of the SEIQR model over the interval \([0, 100]\). Panels (a) through (e) show the dynamics of the susceptible \(S(t)\), exposed \(E(t)\), infectious \(I(t)\), quarantined \(Q(t)\), and recovered \(R(t)\) populations, respectively.}
\label{fig: Extinction}
\end{figure}

\newpage

\begin{figure}[H]
\centering
\subfloat[]{\includegraphics[width=0.47\textwidth]{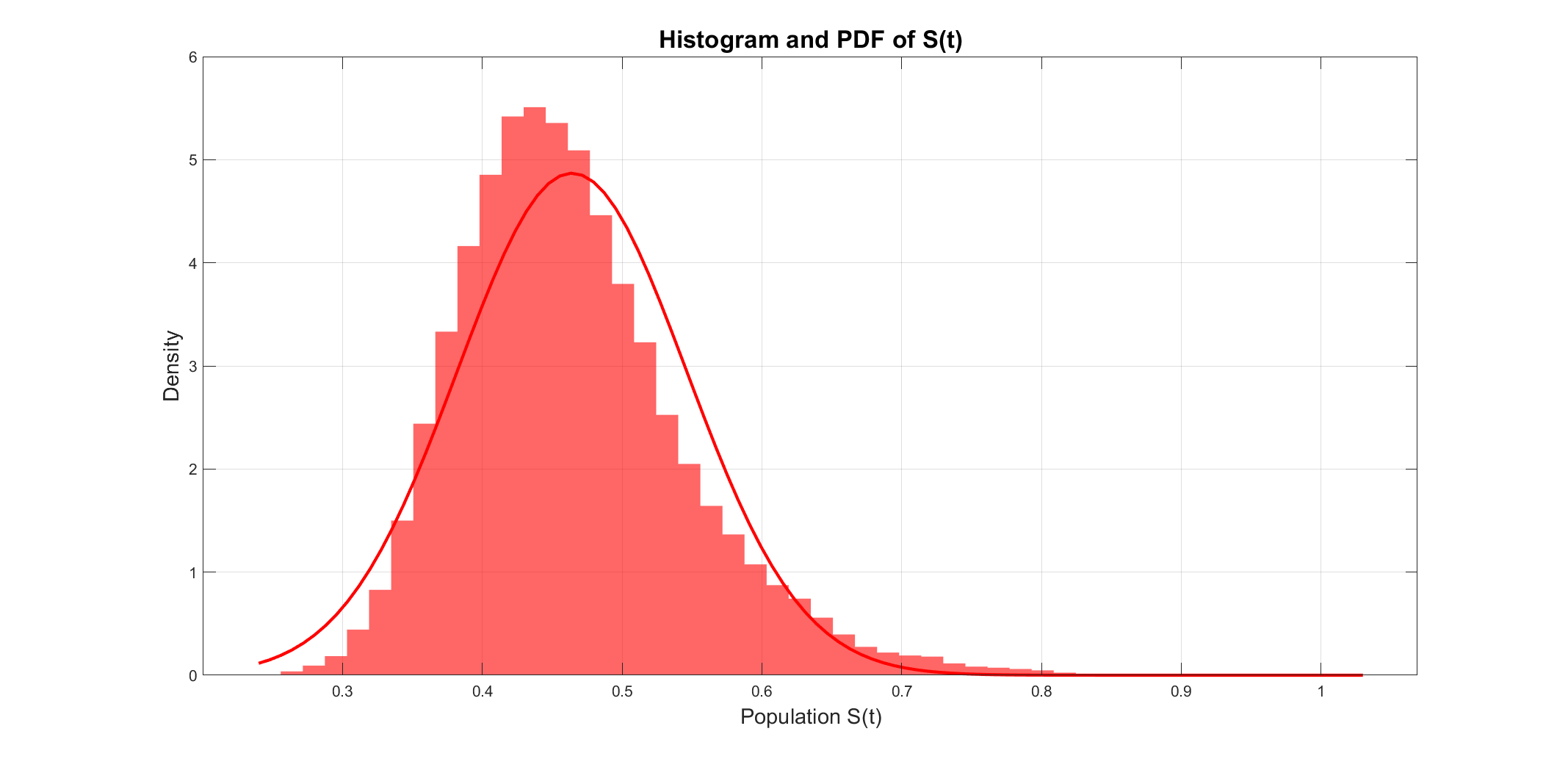}}\qquad
\subfloat[]{\includegraphics[width=0.47\textwidth]{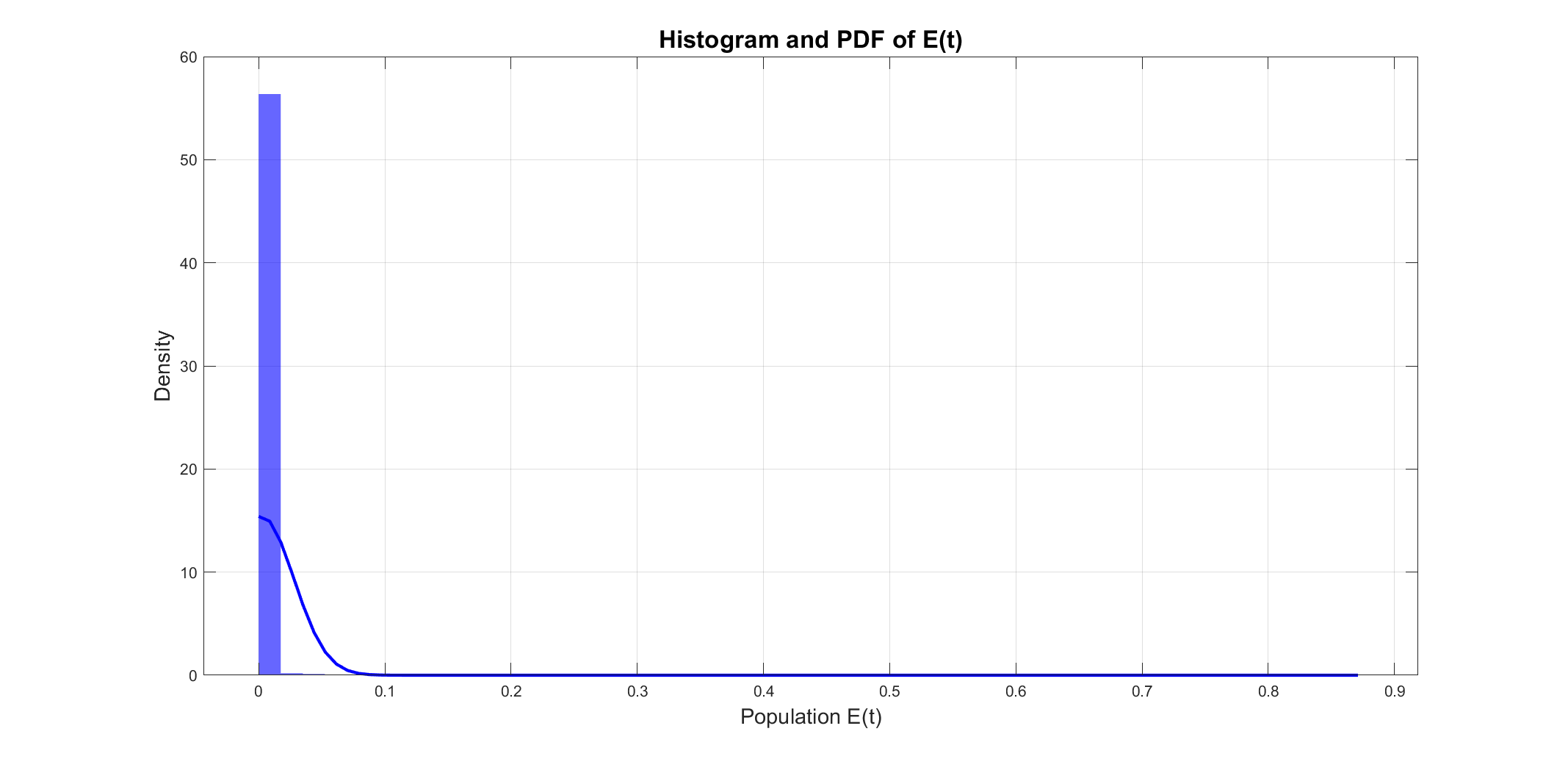}}\qquad
\subfloat[]{\includegraphics[width=0.47\textwidth]{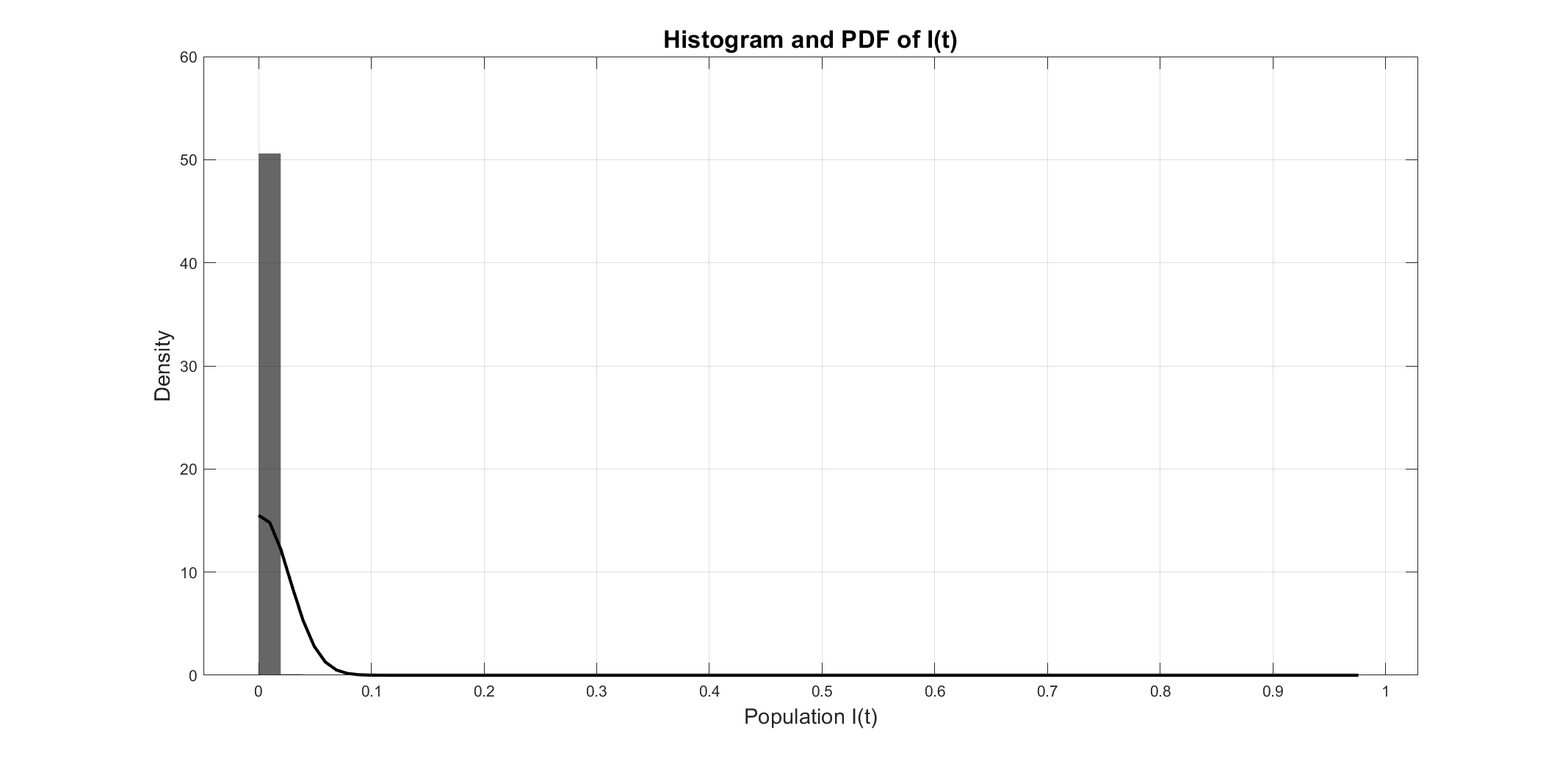}}\qquad
\subfloat[]{\includegraphics[width=0.47\textwidth]{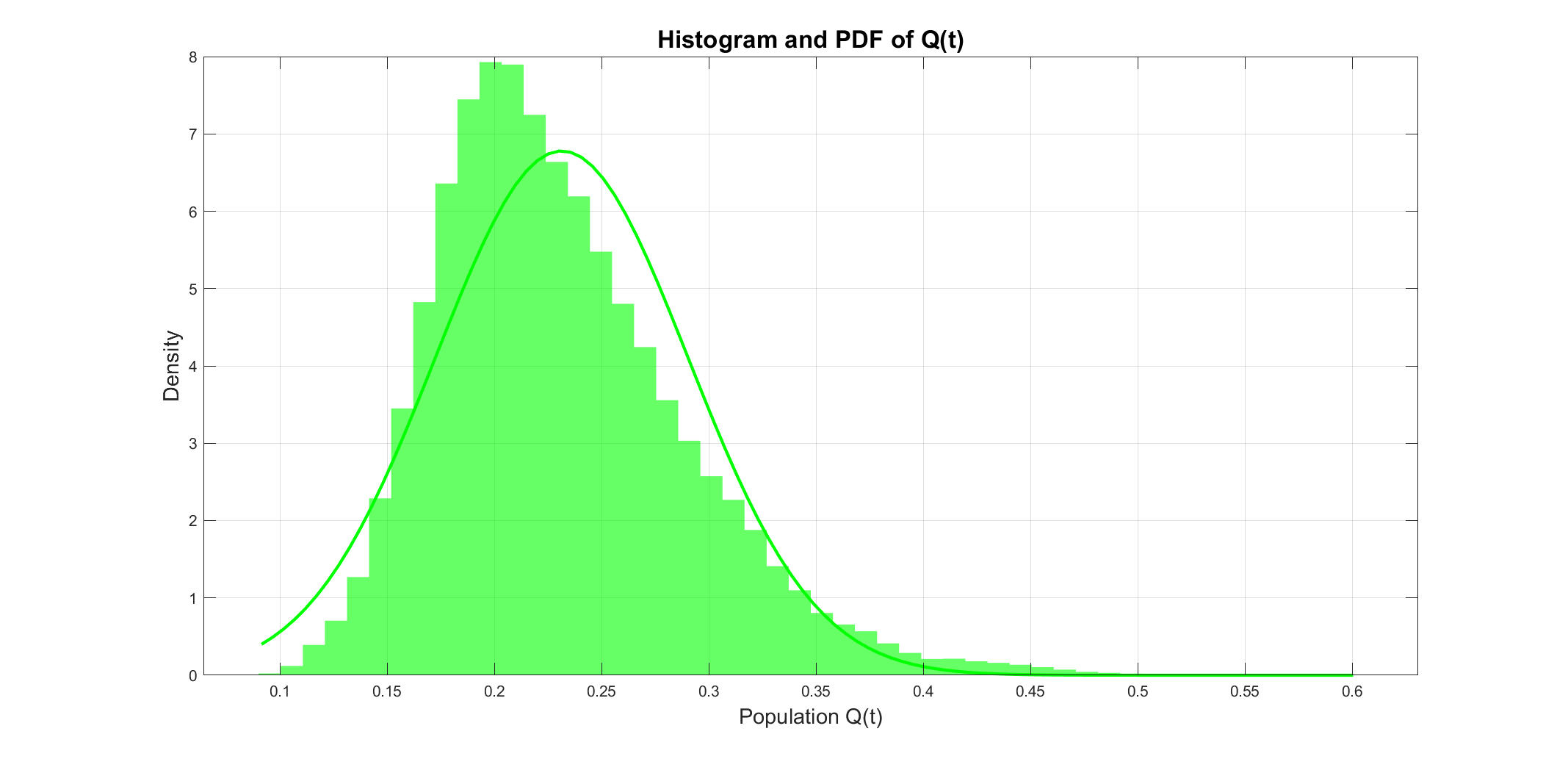}}\qquad
\subfloat[]{\includegraphics[width=0.55\textwidth]{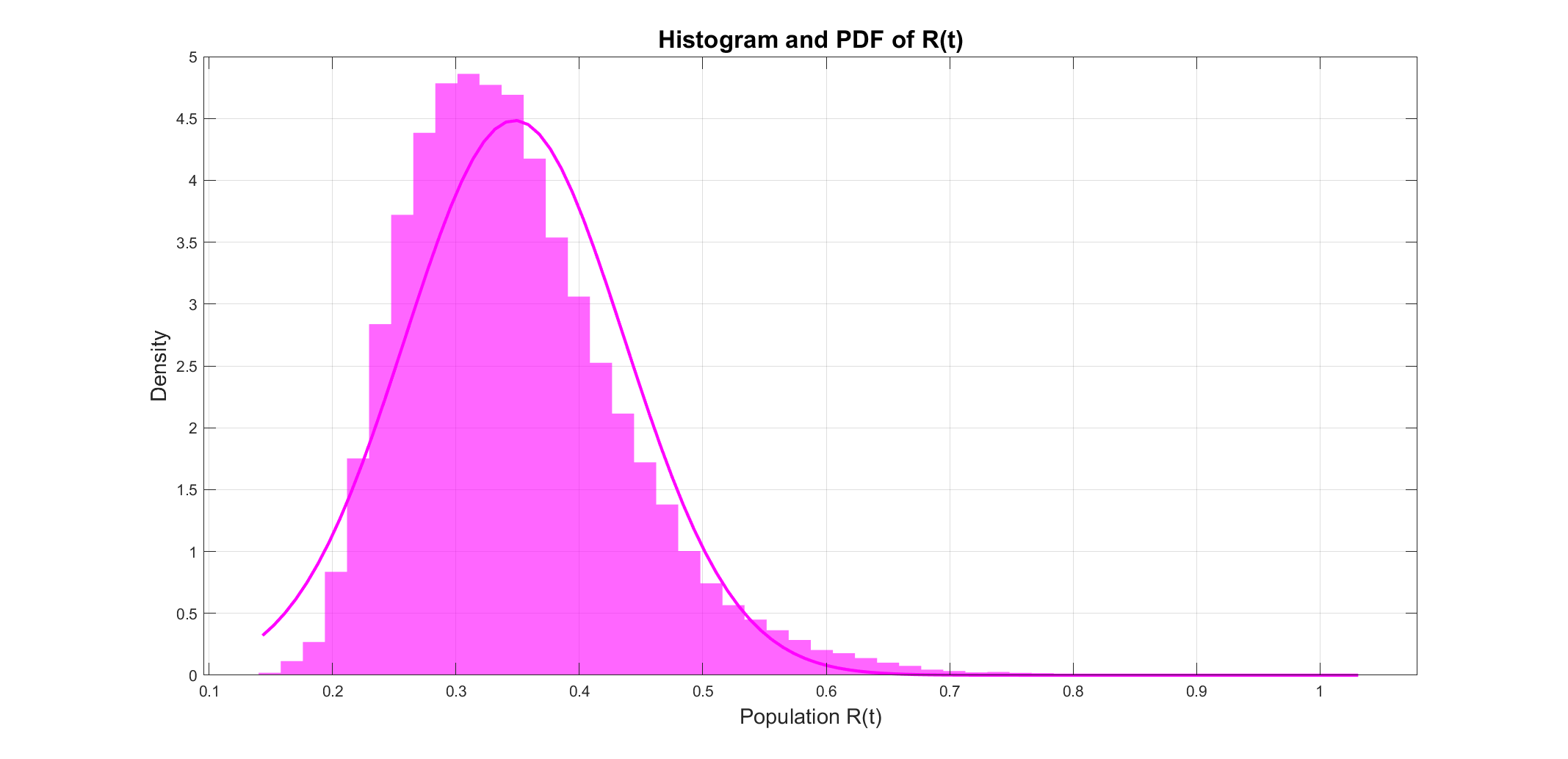}}

\caption{Histograms and the associated theoretical probability density functions of the stochastic paths.}
\label{fig: Simpdf}
\end{figure}

The plots in Figure \ref{fig: Simpdf}  show the histograms for the stochastic paths of the compartments \( S(t) \), \( E(t) \), \( I(t) \), \( Q(t) \), and \( R(t) \) at \( t = 1000 \), along with their theoretical probability density functions. These visualizations demonstrate how well the simulation results align with theoretical predictions and reveal the distribution of each compartment. For large \( t \), the results indicate that some compartments approach extinction, highlighting the impact of stochastic dynamics on the system's long-term behavior.

\begin{figure}[H]
\centering
\subfloat[]{\includegraphics[width=0.45\textwidth]{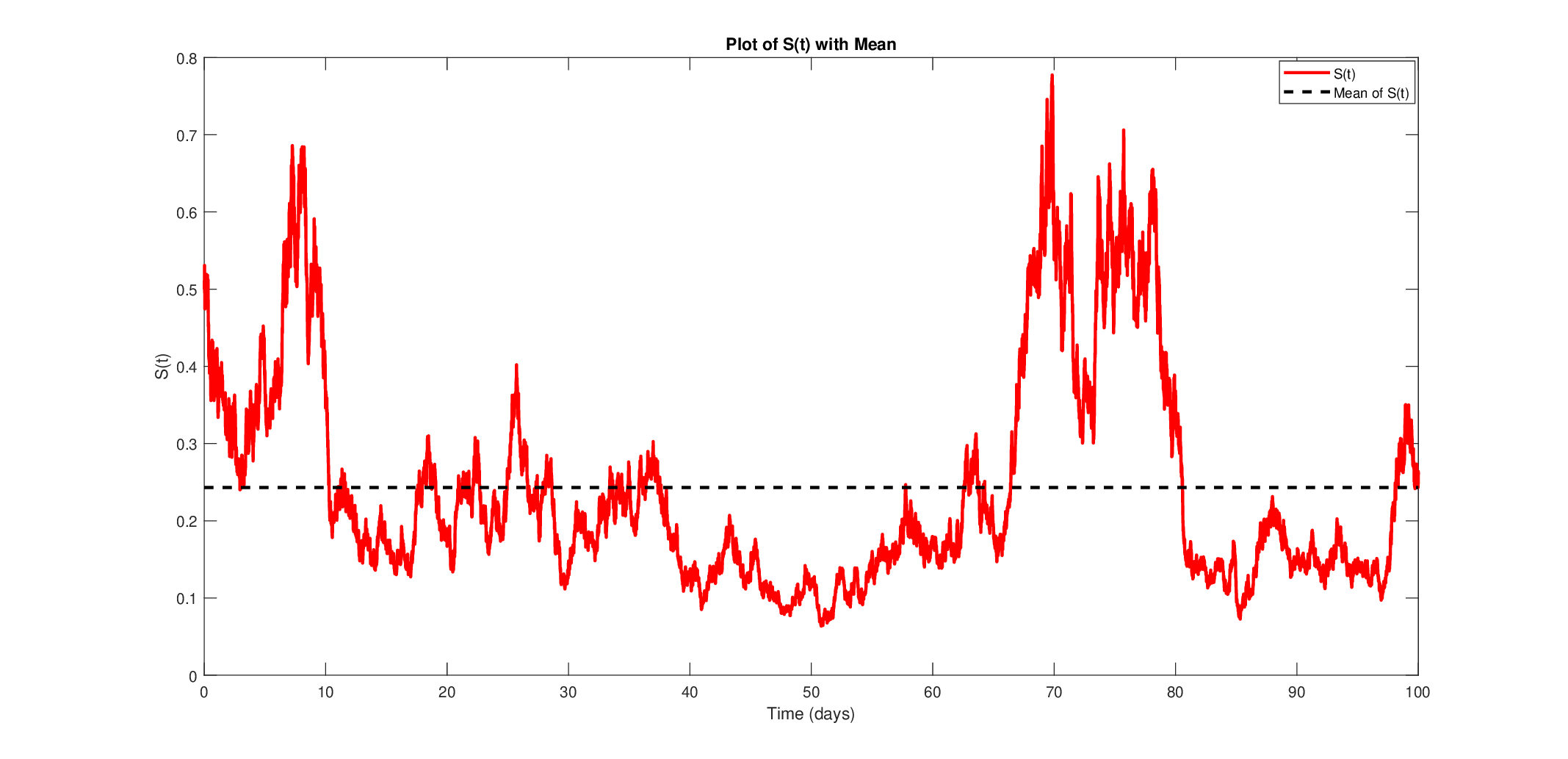}}\qquad
\subfloat[]{\includegraphics[width=0.45\textwidth]{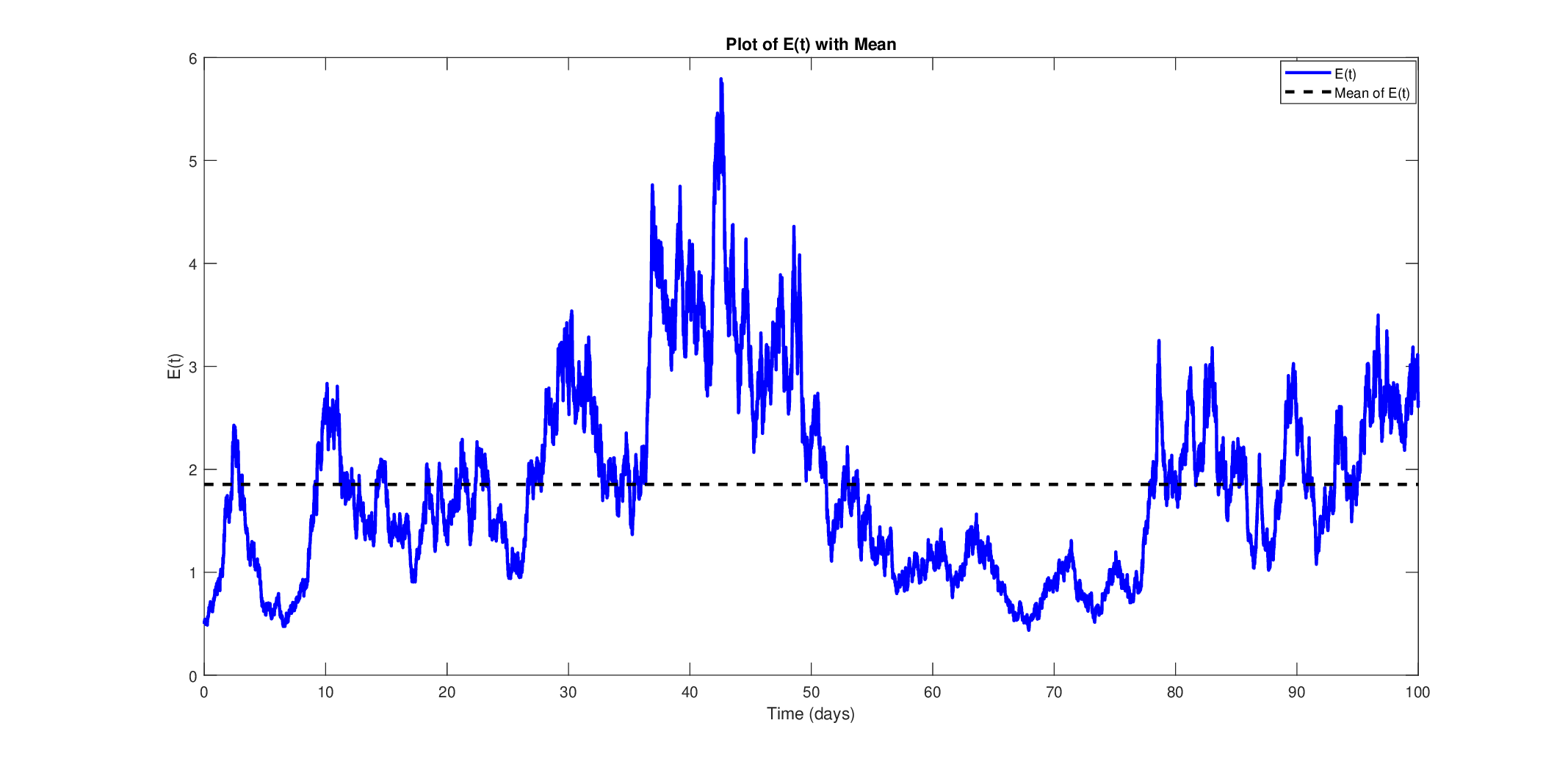}}\qquad
\subfloat[]{\includegraphics[width=0.45\textwidth]{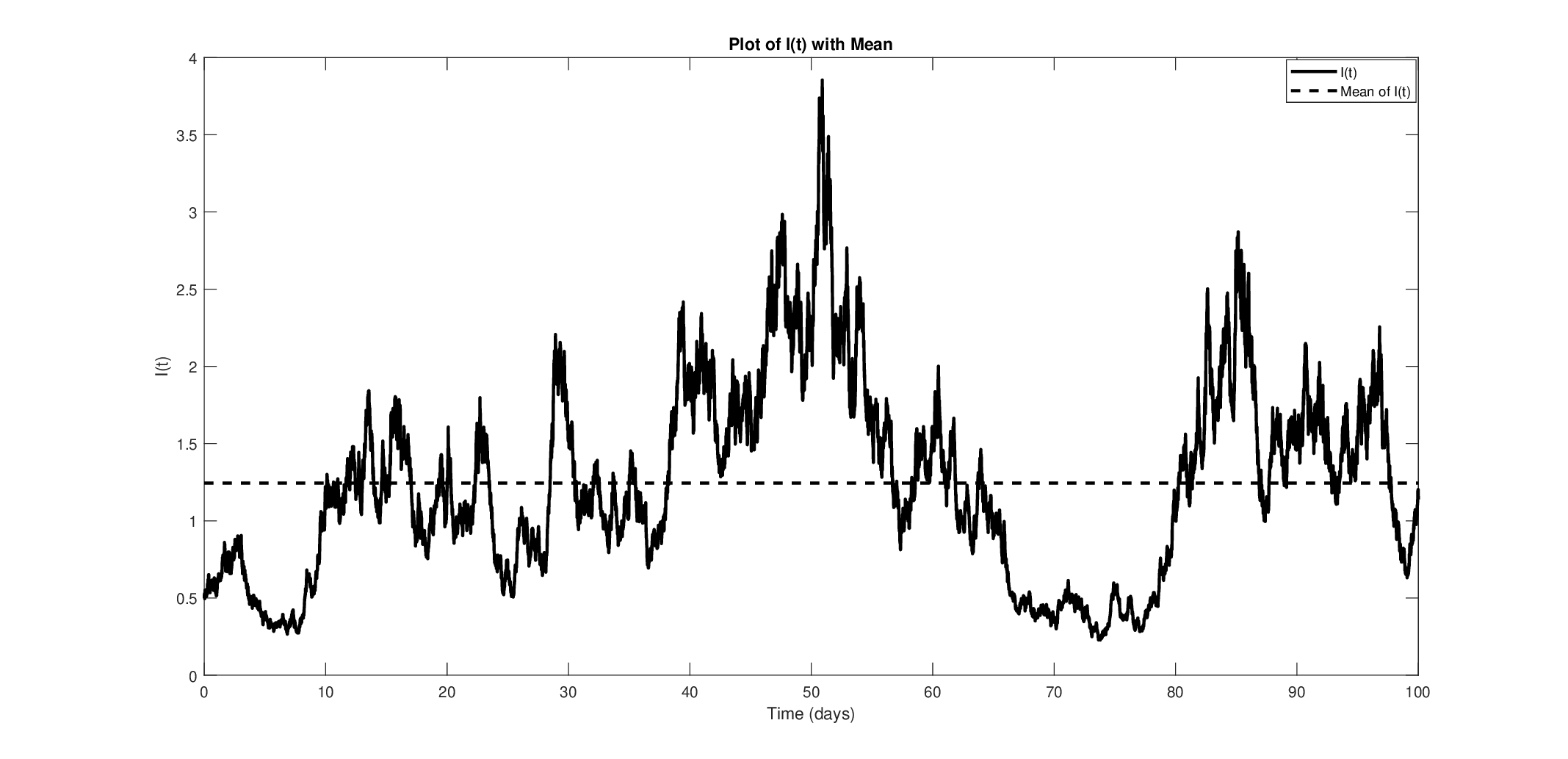}}\qquad
\subfloat[]{\includegraphics[width=0.45\textwidth]{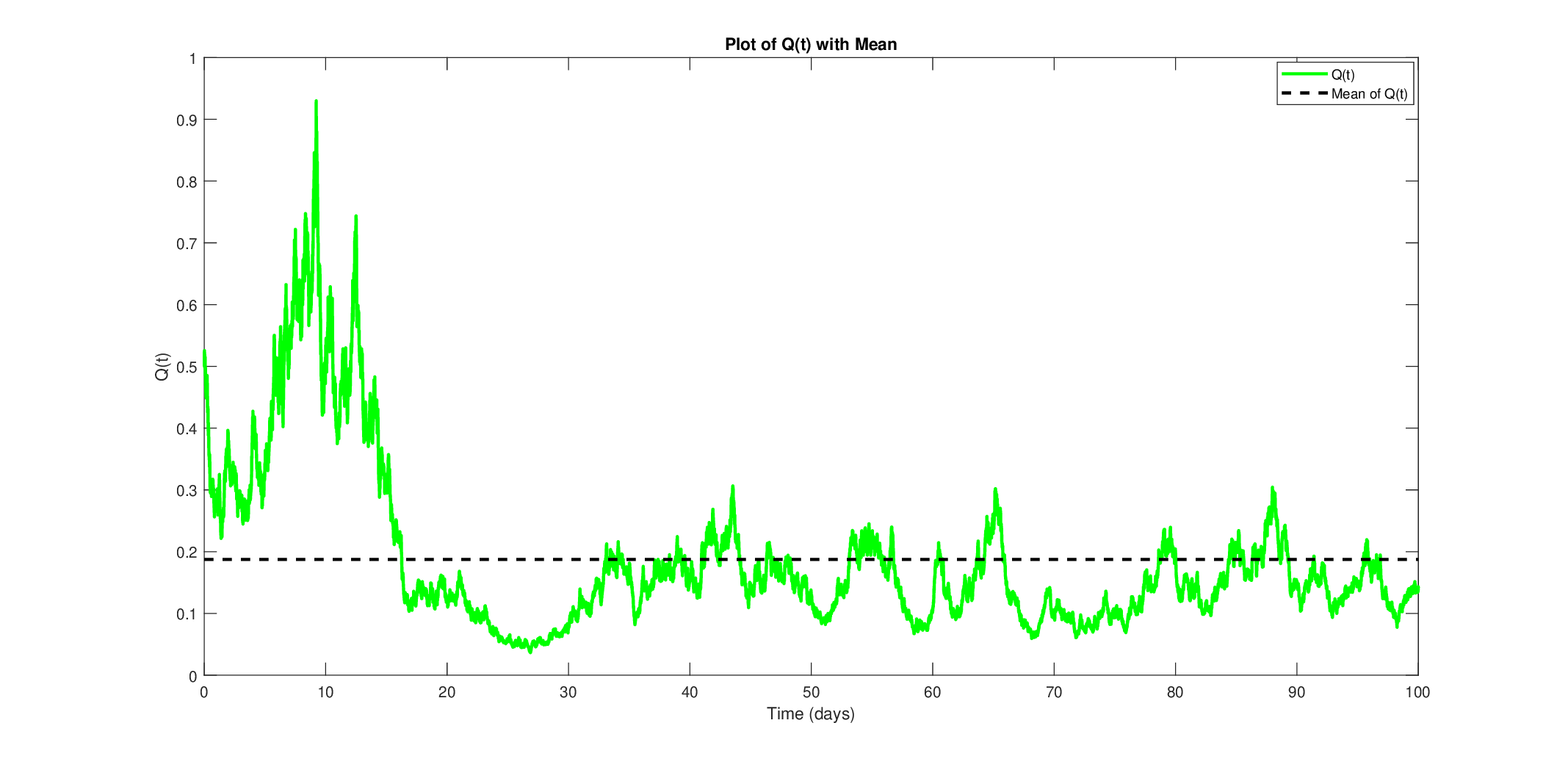}}\qquad
\subfloat[]{\includegraphics[width=0.45\textwidth]{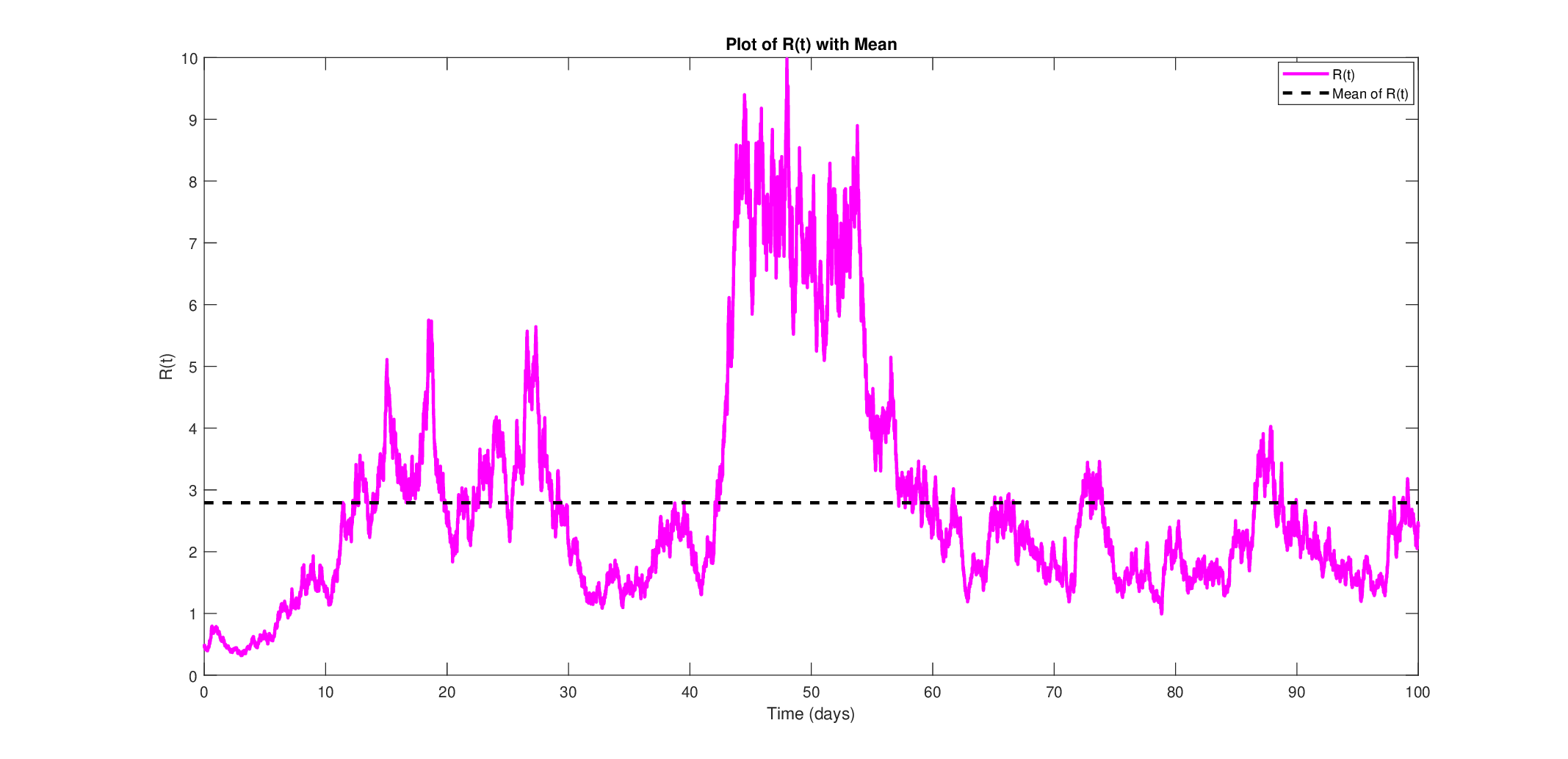}}

\caption{ Trajectories of the SEIQR model over the interval $[0, 100]$. (a) $S(t)$; (b) $E(t)$; (c) $I(t)$; (d) $Q(t)$; (e) $R(t)$.}
\label{fig: Persistence}
\end{figure}



 Moreover, the mesh surface plots in Figure \ref{fig: SimMesh1} and \ref{fig: SimMesh2} illustrate the joint probability density functions of different pairs of stochastic variables at time $t = 1000$. The plots present the joint distributions of the following pairs: \( S(t) \) and \( E(t) \), \( S(t) \) and \( I(t) \), \( S(t) \) and \( Q(t) \), \( S(t) \) and \( R(t) \), \( E(t) \) and \( I(t) \), \( E(t) \) and \( Q(t) \), \( E(t) \) and \( R(t) \), \( I(t) \) and \( Q(t) \), \( I(t) \) and \( R(t) \), and \( Q(t) \) and \( R(t) \). Each plot features a $3D$ surface view of the joint density, with the \( x \)- and \( y \)-axes representing the variables of interest and the \( z \)-axis indicating the density. These plots effectively reveal the interactions between different compartments in the SEIQR model, highlighting areas of higher joint probability where the densities are significant. This detailed visualization aids in understanding how the compartments influence each other and the overall dynamics of the epidemic model.

\begin{figure}[H]
\centering
\subfloat[]{\includegraphics[width=0.47\textwidth]{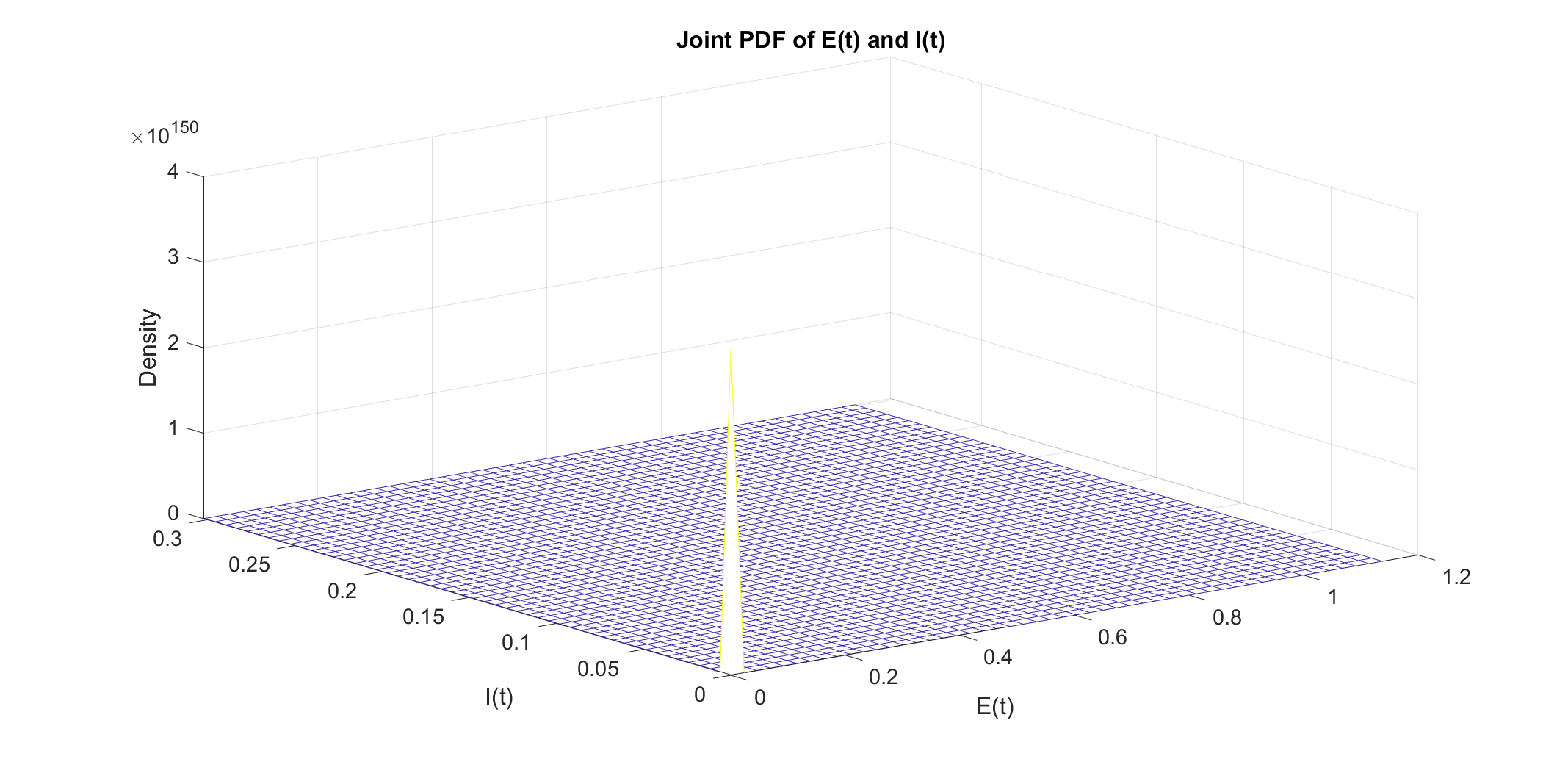}}\qquad
\subfloat[]{\includegraphics[width=0.47\textwidth]{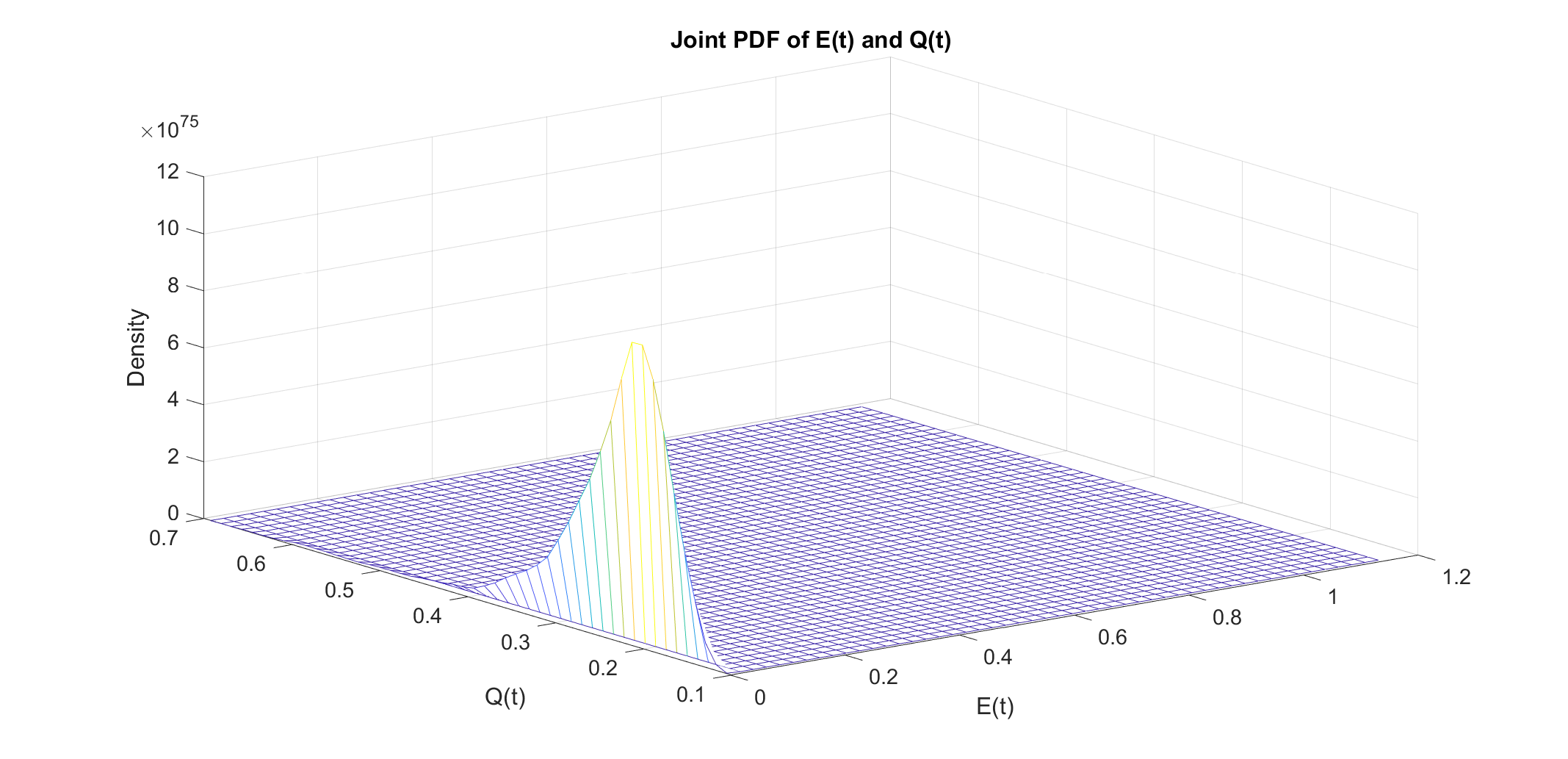}}\qquad
\subfloat[]{\includegraphics[width=0.47\textwidth]{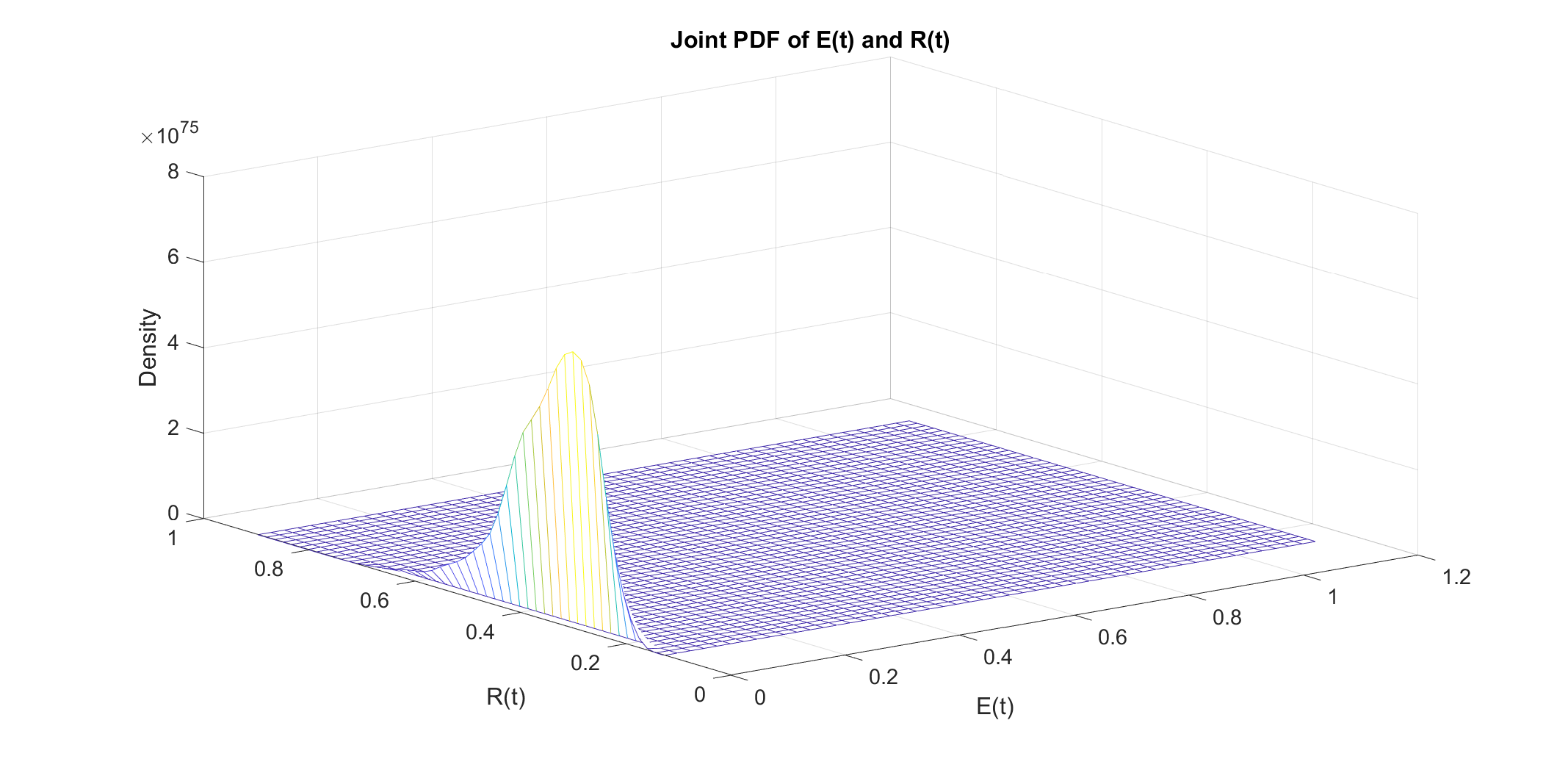}}\qquad
\subfloat[]{\includegraphics[width=0.47\textwidth]{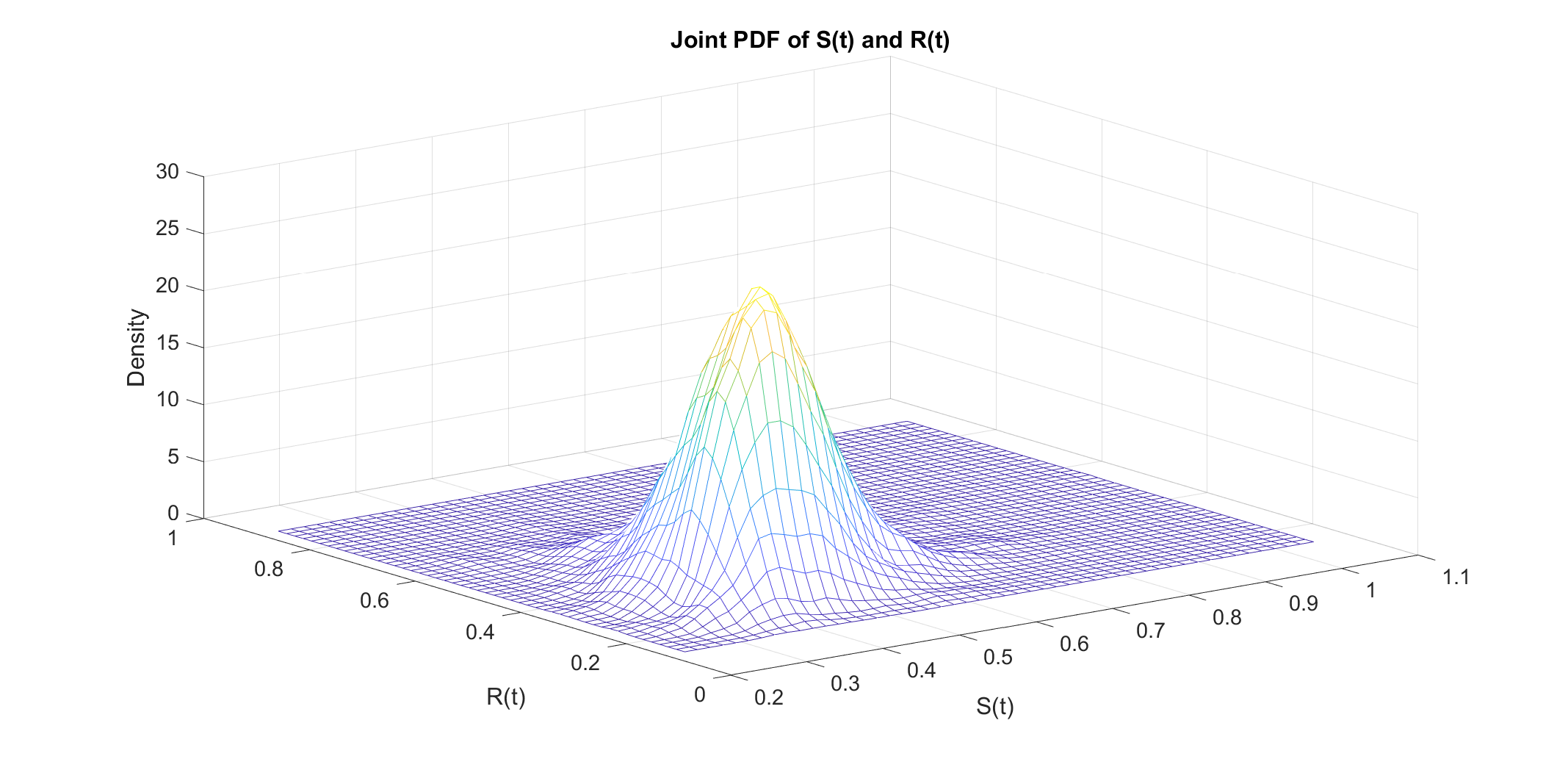}}\qquad
\subfloat[]{\includegraphics[width=0.55\textwidth]{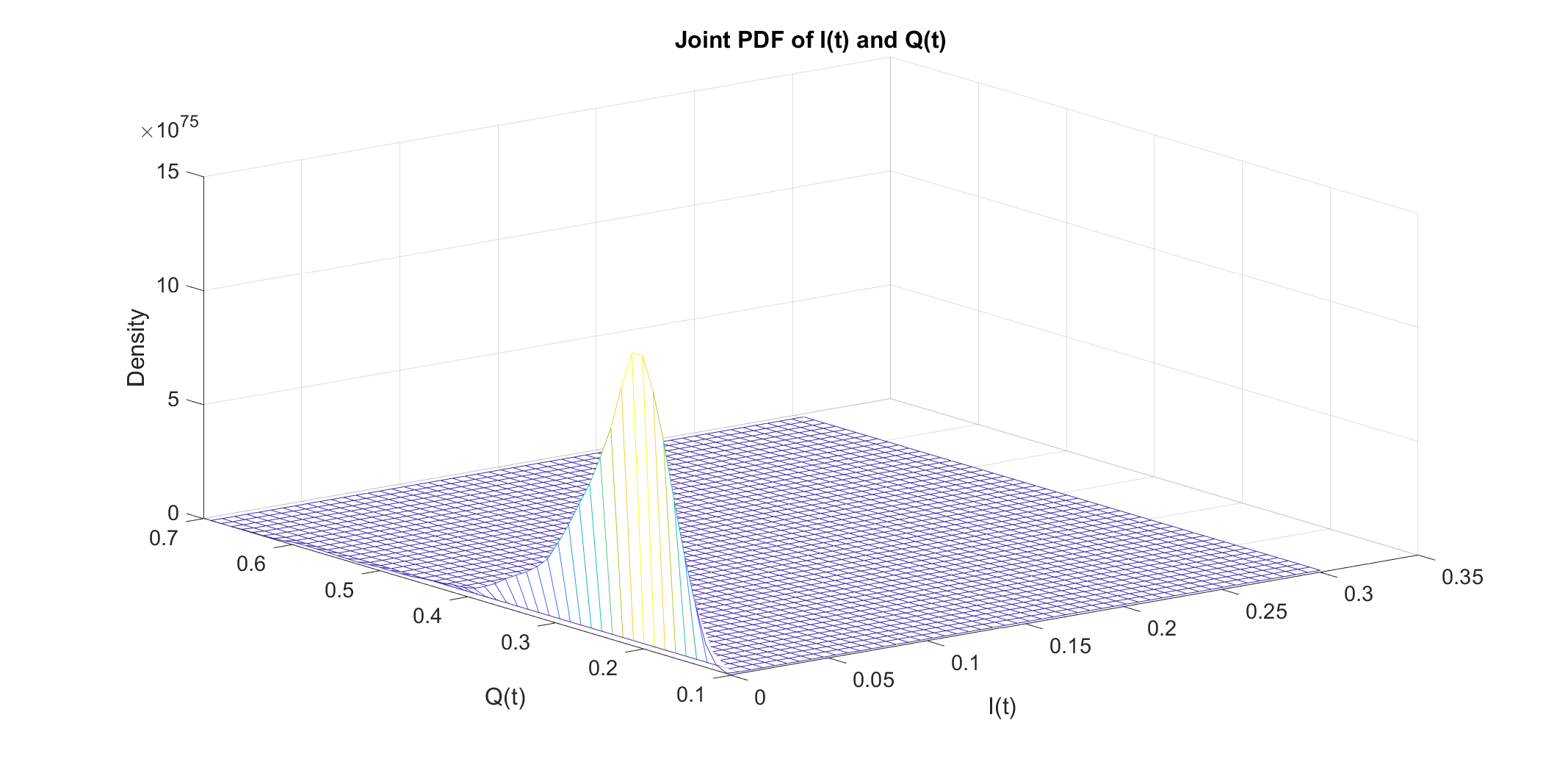}}

\caption{Numerical mesh surface plots of the joint probability density functions for various pairs of stochastic variables at \( t = 1000 \). (a) \( E(t) \) vs. \( I(t) \); (b) \( E(t) \) vs. \( Q(t) \); (c) \( E(t) \) vs. \( R(t) \); (d) \( S(t) \) vs. \( R(t) \); (e) \( I(t) \) vs. \( Q(t) \). }
\label{fig: SimMesh1}
\end{figure}

\begin{figure}[H]
\centering
\subfloat[]{\includegraphics[width=0.47\textwidth]{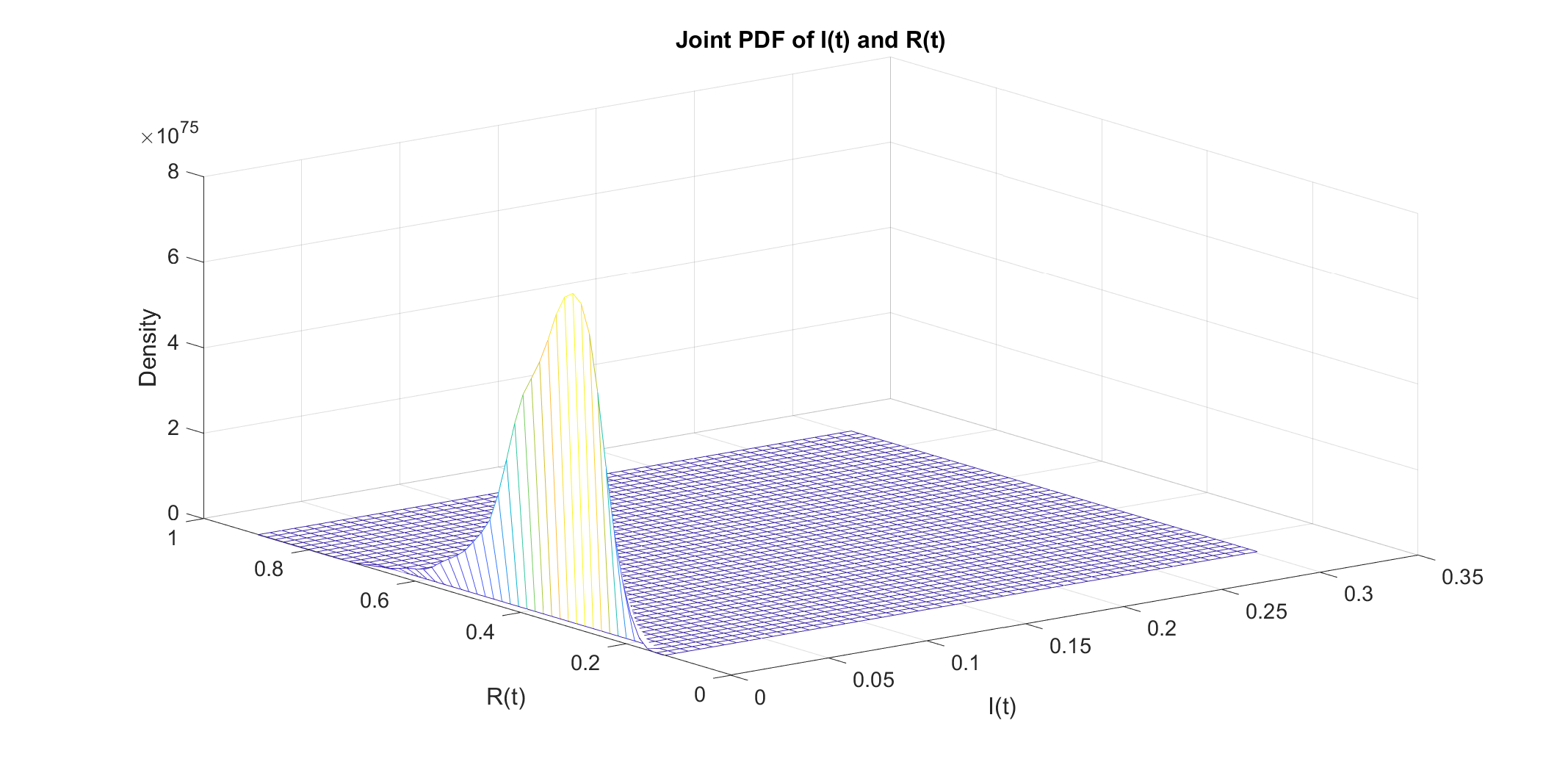}}\qquad
\subfloat[]{\includegraphics[width=0.47\textwidth]{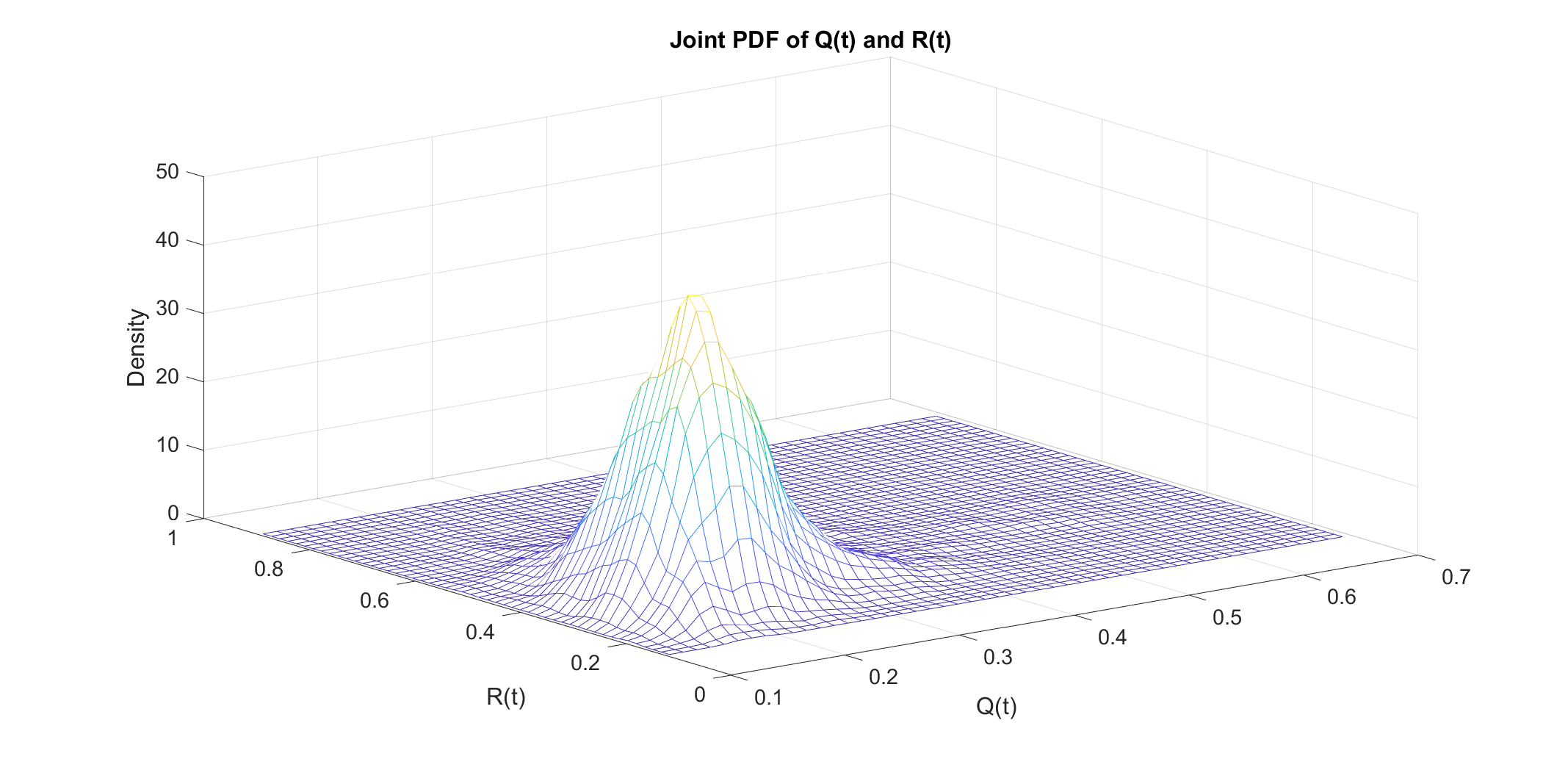}}\qquad
\subfloat[]{\includegraphics[width=0.47\textwidth]{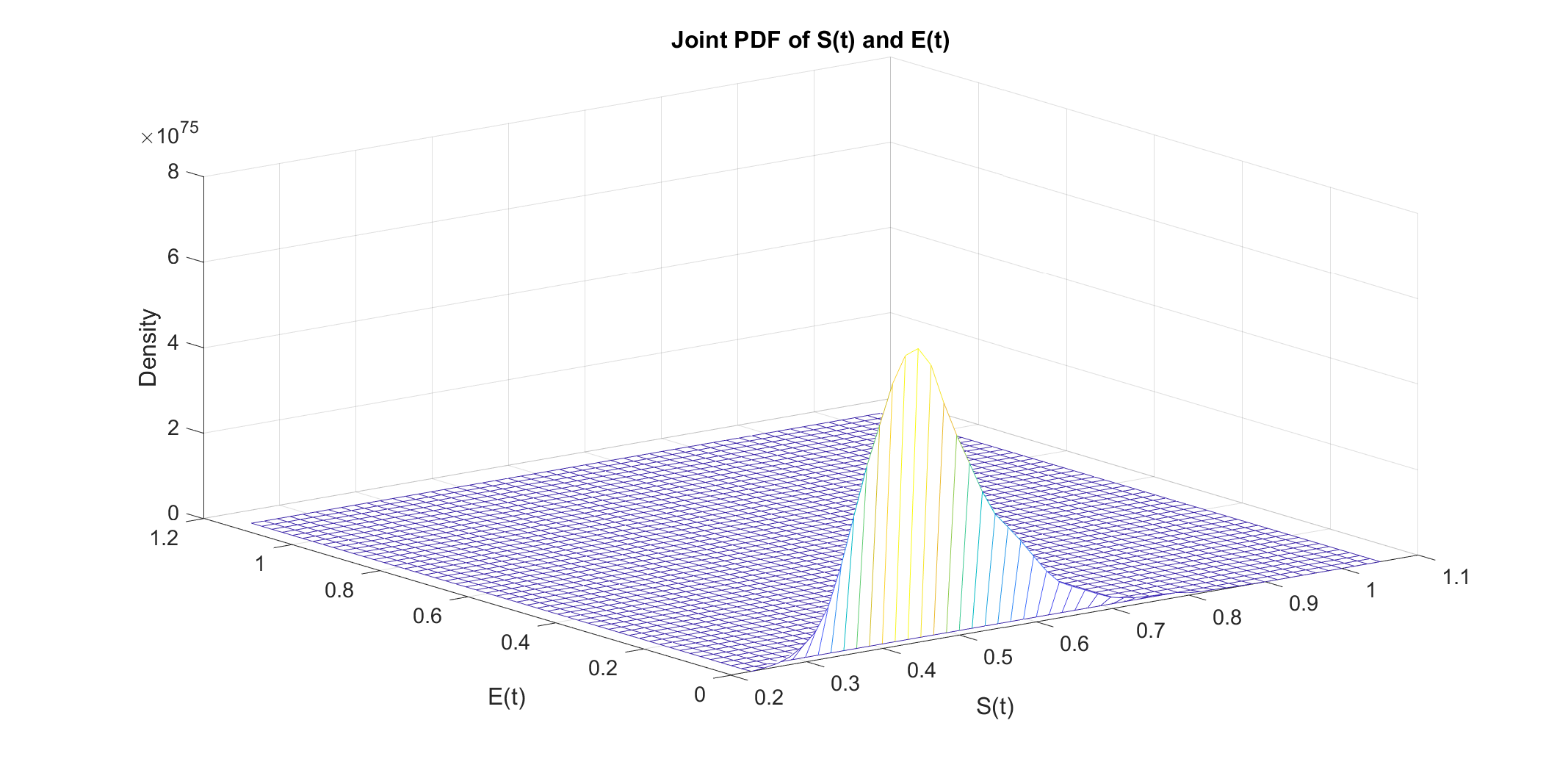}}\qquad
\subfloat[]{\includegraphics[width=0.47\textwidth]{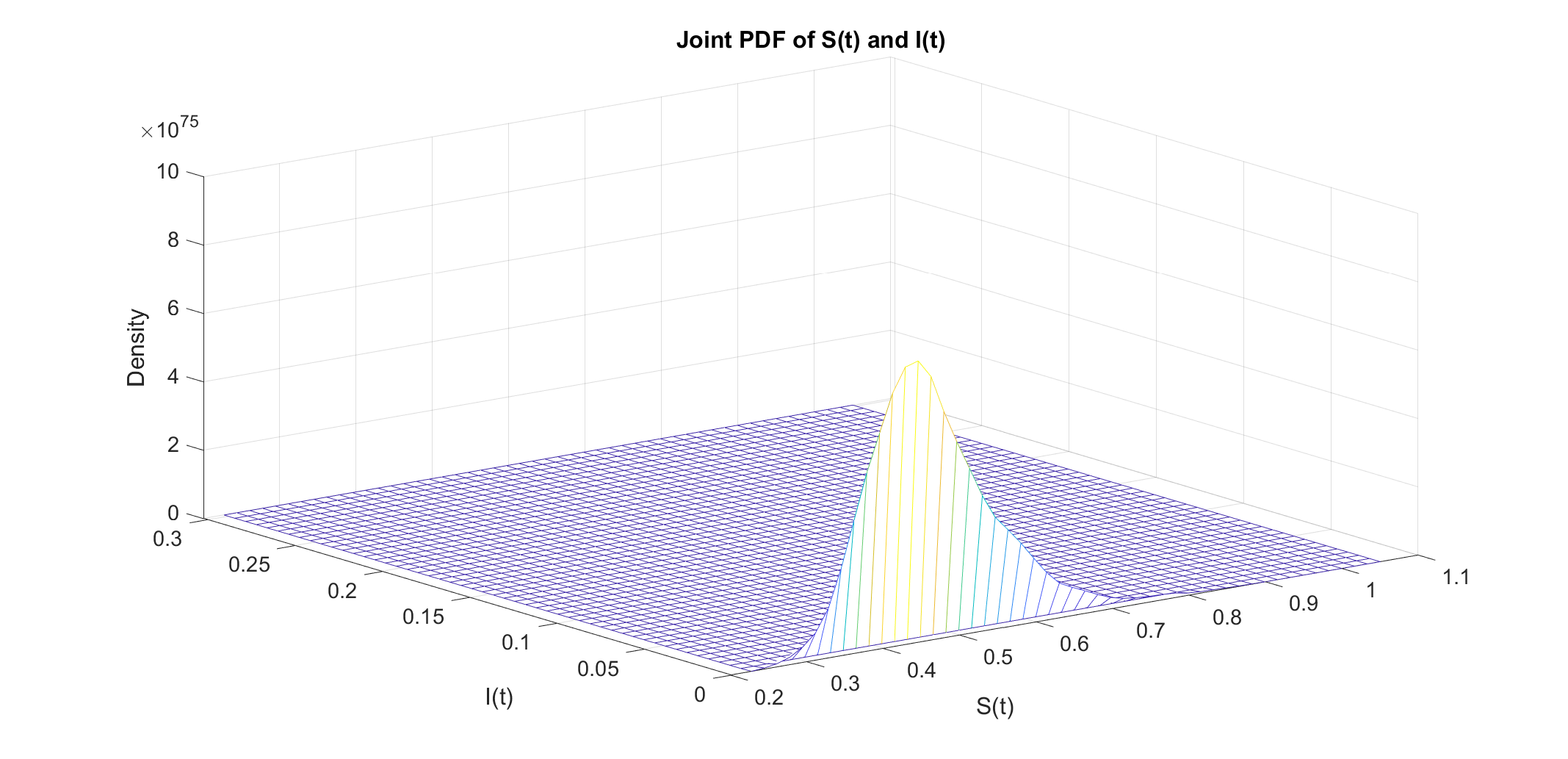}}\qquad
\subfloat[]{\includegraphics[width=0.55\textwidth]{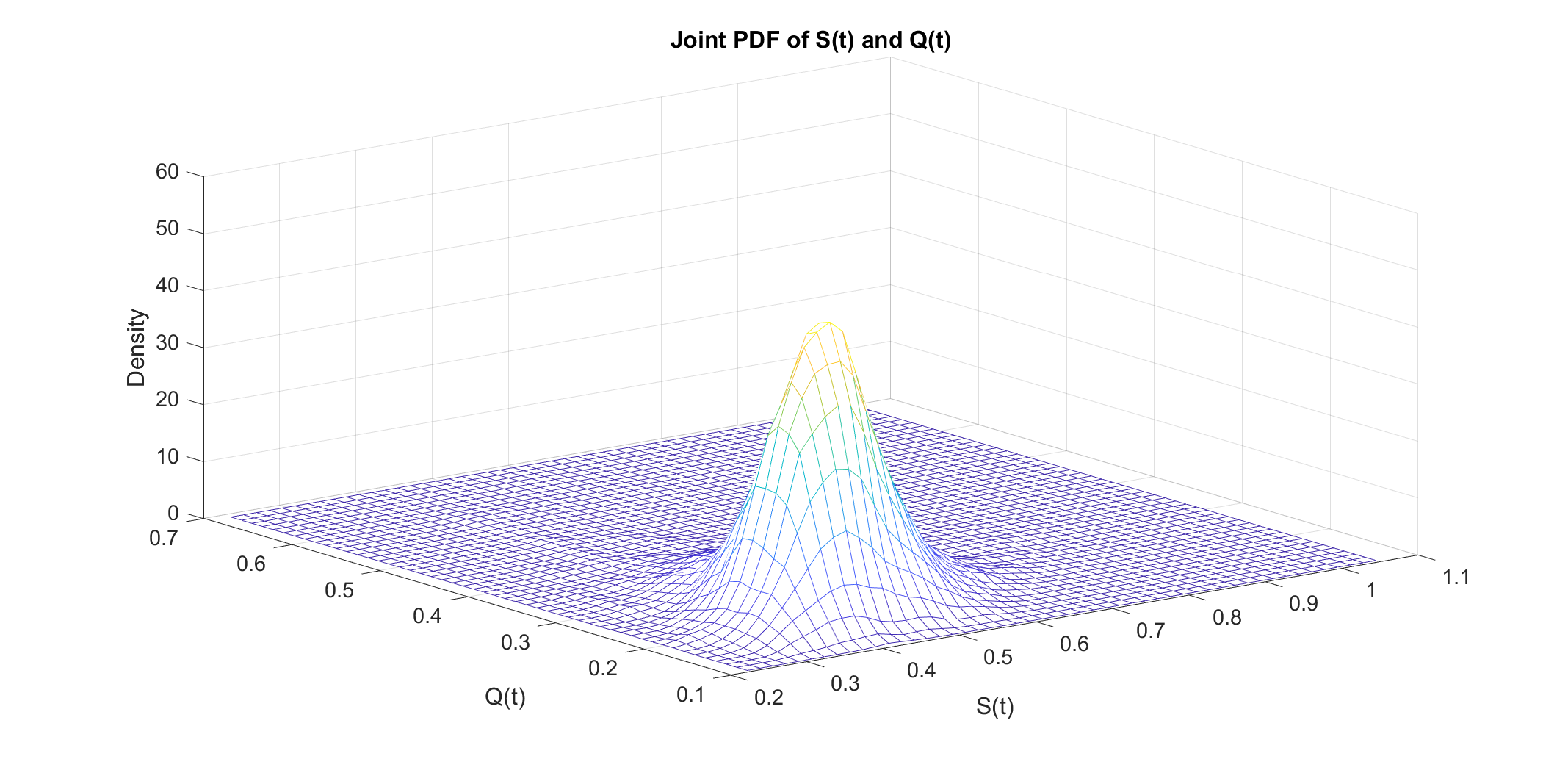}}

\caption{Numerical mesh surface plots of the joint probability density functions for various pairs of stochastic variables at \( t = 1000 \). (a) \( I(t) \) vs. \( R(t) \); (b) \( Q(t) \) vs. \( R(t) \); (c) \( S(t) \) vs. \( E(t) \); (d) \( S(t) \) vs. \( I(t) \); (e) \( S(t) \) vs. \( Q(t) \).}
\label{fig: SimMesh2}
\end{figure}

\section{Conclusion}
In this paper we propose a stochastic SEIQR epidemic model with a generalized incidence rate. We proved the uniqueness and the boundedness of the global positive solution for any positive initial value.
We also derived the stochastic ultimate boundedness and permanence for the system \eqref{sys1}. Moreover, we investigate the stochastic extinction and V-geometric ergodicity. All these theoretical findings are verified using the first-order Itô-Taylor stochastic scheme (also called the Milstein scheme).

{\small

%

%
%
%
%
%
%
\end{document}